\documentclass[11pt, a4paper]{article}
\pdfoutput=1

\usepackage{amsfonts, amsmath, amssymb, amsthm, booktabs, dsfont, enumitem, multicol, multirow, rotating, setspace, xcolor}
\usepackage[english]{babel}
\usepackage[font={small}]{caption}
\usepackage[T1]{fontenc}
\usepackage[left=2.5cm, right=2.5cm, top=2.5cm, bottom=2.5cm, includefoot]{geometry}
\usepackage[breaklinks, colorlinks=true, pagebackref=false, pdfauthor={}, pdftitle={Triptych}, citecolor=blue, linkcolor=blue, urlcolor=blue]{hyperref}
\usepackage{cleveref}
\usepackage[utf8]{inputenc}
\usepackage[authoryear,round]{natbib} 
\usepackage[section]{placeins}
\usepackage[font={small}]{subcaption}
 \usepackage[textsize=scriptsize]{todonotes}

\setlist[enumerate]{align=left}
\numberwithin{equation}{section}

\newtheorem{theorem}{Theorem}
\newtheorem{fact}{Fact} 
\newtheorem{lemma}[theorem]{Lemma}

\newcommand{\myE}{\mathbb{E}}

\newcommand{\myQ}{\mathbb{Q}}
\newcommand{\CEP}{\text{CEP}}

\newcommand{\HR}{\text{HR}}
\newcommand{\FAR}{\text{FAR}}

\newcommand{\AUC}{\text{AUC}}

\newcommand{\myS}{\textsf{S}}
\newcommand{\BS}{\textsf{BS}}
\newcommand{\LogS}{\textsf{LogS}}

\newcommand{\SX}{\bar{\myS}}
\newcommand{\SC}{\bar{\myS}_{\textsf{C}}}
\newcommand{\SR}{\bar{\myS}_{\textsf{R}}}

\newcommand{\SXtheta}{\bar{\myS}_\theta}

\newcommand{\MCB}{\text{MCB}}
\newcommand{\DSC}{\text{DSC}}
\newcommand{\UNC}{\text{UNC}}

\title{Evaluating Probabilistic Classifiers: The Triptych}

\author{Timo Dimitriadis\thanks{Alfred Weber Institute of Economics, Heidelberg University and Computational Statistics (CST) group, Heidelberg Institute for Theoretical Studies, Germany; e-mail: \texttt{timo.dimitriadis@awi.uni-heidelberg.de}} 
\and Tilmann Gneiting\thanks{Computational Statistics (CST) group, Heidelberg Institute for Theoretical Studies, and Institute for Stochastics, Karlsruhe Institute of Technology, Germany; e-mail: \texttt{tilmann.gneiting@h-its.org}} 
\and Alexander I. Jordan\thanks{Computational Statistics (CST) group, Heidelberg Institute for Theoretical Studies, Germany; e-mail: \texttt{alexander.jordan@h-its.org}} 
\and Peter Vogel\thanks{\fussy Institute for Stochastics, Karlsruhe Institute of Technology, Germany; e-mail: \texttt{petervogel1991@googlemail.com}} 
}

\begin{document}

\maketitle

\begin{abstract}
Probability forecasts for binary outcomes, often referred to as probabilistic classifiers or confidence scores, are ubiquitous in science and society, and methods for evaluating and comparing them are in great demand.  We propose and study a triptych of diagnostic graphics that focus on distinct and complementary aspects of forecast performance: The reliability diagram addresses calibration, the receiver operating characteristic (ROC) curve diagnoses discrimination ability, and the Murphy diagram visualizes overall predictive performance and value.  A Murphy curve shows a forecast's mean elementary scores, including the widely used misclassification rate, and the area under a Murphy curve equals the mean Brier score.  For a calibrated forecast, the reliability curve lies on the diagonal, and for competing calibrated forecasts, the ROC and Murphy curves share the same number of crossing points.  We invoke the recently developed CORP (Consistent, Optimally binned, Reproducible, and Pool-Adjacent-Violators (PAV) algorithm based) approach to craft reliability diagrams and decompose a mean score into miscalibration (\MCB), discrimination (\DSC), and uncertainty (\UNC) components.  Plots of the $\DSC$ measure of discrimination ability versus the calibration metric $\MCB$ visualize classifier performance across multiple competitors.  The proposed tools are illustrated in empirical examples from astrophysics, economics, and social science.
\end{abstract}


\vfill
\section{Introduction}  \label{sec:introduction}

Across science and society, probability forecasts for the occurrence of a binary outcome, also referred to as probabilistic classifiers or confidence scores, are widely used.  Prominent examples include a patient's recovery or survival, weather events, solar flares, the designation of email as spam, credit approval, and recidivism of criminal defendants, to name but a few applications.  Evidently, our ability to develop and improve probability forecasts depends on the availability of diagnostic tools for the assessment and comparison of predictive power.  \vfill \pagebreak

While some applications call for the use of a single numerical performance measure, with forecast contests and leader boards being prime examples, the condensation of forecast quality into a single number prevents detailed analyses.  As \citet{Janssens2020} notes, 
\begin{quote}
\footnotesize
``Some prediction researchers prefer one metric or graph that captures the overall performance of prediction models.  Others prefer one for each different aspect of performance, such as calibration, discrimination, predictive value (risks), and utility.'' 
\end{quote} 
Not surprisingly, numerous types of diagnostic graphics for the evaluation of probability forecasts exist \citep{Murphy1992, Prati2011, Filho2021}, and practitioners may wonder which ones are to be preferred.

In this article we propose the use of a triptych of diagnostic graphics and provide theoretical support for our choices.  The triptych consists of the reliability diagram in the recently proposed CORP (Consistent, Optimally binned, Reproducible, and Pool-Adjacent-Violators (PAV) algorithm based) form to assess calibration \citep{Dimitriadis2021}, the receiver operating characteristic (ROC) curve to judge discrimination ability \citep{Swets1973, Gneiting2022a}, and the Murphy diagram for the assessment of overall predictive performance and utility \citep{Ehm2016}.  Figure \ref{fig:C1_triptych} illustrates the triptych for probabilistic classifiers from an astrophysical forecast challenge \citep{Leka2019a, Leka2019b} as introduced in Table \ref{tab:C1} and discussed in detail in Section \ref{sec:solar}. 

\begin{table}[t]
\centering
\footnotesize
\caption{Probability forecasts for class C1.0+ solar flares at a prediction horizon of a day ahead from a joint test set within calendar years 2016 and 2017 \citep{Leka2019a, Leka2019b}: Acronym, source, mean Brier score, mean logarithmic (Log) score, and misclassification rate (MR).  Details of the data example are discussed in Section \ref{sec:solar}.}  \label{tab:C1}		
\begin{tabular}{llrccc}
\toprule
\multicolumn{2}{c}{Probability Forecast} && \multicolumn{3}{c}{Mean Score}  \\
\cmidrule{1-2} \cmidrule{4-6} 
Acronym & Source                                          && Brier & Log & MR \\
\cmidrule{1-2} \cmidrule{4-6} 
NOAA    & National Oceanic and Atmospheric Administration && 0.144 & 0.449    & 0.205 \\ 
SIDC    & Royal Observatorium Belgium                     && 0.172 & 0.515    & 0.263 \\  
ASSA    & Korean Space Weather Agency                     && 0.184 & $\infty$ & 0.273 \\ 
MCSTAT  & Trinity College Dublin                          && 0.193 & 0.587    & 0.275 \\
\bottomrule
\end{tabular}
\end{table}

\begin{figure}[t]
\centering
\includegraphics[width=\linewidth]{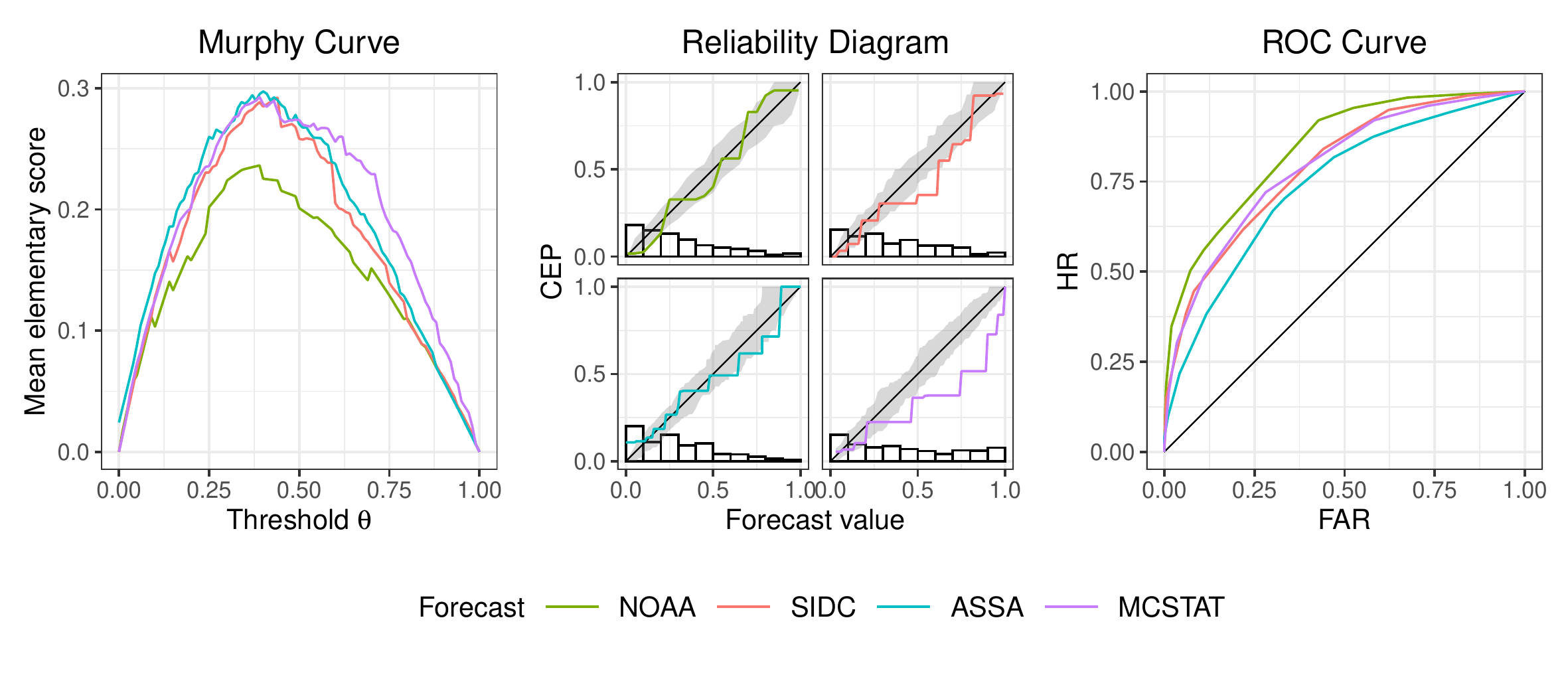}
\vspace{-11mm}
\caption{Triptych of diagnostic graphics for evaluating and comparing the probability forecasts of class C1.0+ solar flares from Table \ref{tab:C1}: Murphy curve, reliability diagram, and ROC curve.}  \label{fig:C1_triptych}
\end{figure}

Starting from the left, the Murphy diagram assesses overall predictive performance in terms of proper scoring rules.  To provide background, a scoring rules assigns a score $\myS(x,y)$ to each pair of a probability forecast $x \in [0,1]$ and a binary outcome $y \in \{ 0, 1 \}$, where 1 stands for an event and 0 for a non-event.  A scoring rule is proper if a forecaster minimizes the expected score by issuing a probability forecast that corresponds to her true belief, with the Brier score $\myS(x,y) = (x-y)^2$ and the logarithmic (Log) score $\myS(x,y) = -y \log x - (1-y) \log(1-x)$ being prominent examples \citep{Gneiting2007a}.  Scores then are averaged over a test set, and the forecast with the smallest mean score is considered best.  The widely used misclassification rate (MR) arises as a special case, namely, by assigning a score of 1 if the probability forecast is less than $\frac{1}{2}$ and the event realizes, or the forecast is greater than $\frac{1}{2}$ and the event does not realize, and assigning a score of 0 otherwise.  Distinct proper scoring rules may yield distinct forecast rankings, so practitioners may wonder which one to use, and guidance is essential.  In the case of a binary outcome, any proper scoring rule can be represented as a mixture over so-called elementary scoring rules, and so it suffices to consider only those.  Fortunately, the family of the elementary scoring rules is linearly parameterized by a threshold or cost-loss parameter $\theta$.  In a nutshell, a Murphy curve depicts the mean elementary score as a function of the threshold $\theta$, with lower scores being preferable.  The height of the Murphy curve at $\theta = \frac{1}{2}$ equals the misclassification rate, and the area under the Murphy curve equals the mean Brier score.  If a forecast has a Murphy curve below that of a competitor, then it is superior in terms of any proper scoring rule, and has superior economic utility to any decision maker.  For example, we see from the Murphy diagram in Figure \ref{fig:C1_triptych} that the NOAA forecast dominates the ASSA forecast, regardless of the use intended.  

A probability forecast is calibrated if, conditional on any forecast value $p$, the event realizes in $100 \cdot p$ percent of the instances considered.  Reliability diagrams visualize calibration, by plotting an estimate of the conditional event probability (\CEP) as a function of the forecast value.  While reliability curves close to the diagonal are compatible with assumptions of calibration, notable departures from the diagonal suggest miscalibration and can be interpreted diagnostically.  We adopt the recently proposed CORP approach of \citet{Dimitriadis2021} for the estimation of CEPs by nonparametric isotonic regression, as illustrated in Figure \ref{fig:C1_triptych}, where the SIDC and MCSTAT forecasts exhibit overprediction, with estimates below the diagonal. 

Receiver operating characteristic (ROC) curves visualize the discrimination ability of the forecasts --- that is, they judge to what extent the forecast values distinguish situations with lower or higher true event probabilities.  Specifically, as one issues hard classifiers based on successively higher forecast thresholds, a ROC curve plots the hit rate (HR) on the ordinate against the false alarm rate (FAR) on the abscissa.  As the ROC curve is invariant under strictly increasing transformations of the forecast values, it diagnoses discrimination ability only, while ignoring issues of calibration.  Hit rates close to 1 and false alarm rates close to 0 are desirable, so ROC curves at upper left are indicative of superior discrimination ability.  In the ROC curves in Figure \ref{fig:C1_triptych} the NOAA forecast shows the highest and the ASSA forecast the lowest discrimination ability.  Featuring both excellent calibration and superior discrimination ability, the NOAA forecast also performs best in terms of scoring rules and economic utility, as evidenced by the Murphy curves. 

The choice of the triptych graphics reflects theoretically supported, desirable properties.  Reliability curves exclusively diagnose calibration, ROC curves assess discrimination ability only, and Murphy curves quantify overall predictive performance.  Moreover, the novel Fact \ref{fact:crossingpoints} and Theorem \ref{thm:crossingpoints} below demonstrate that, under perfect calibration, Murphy curves and ROC curves yield congruent insights, as they share the same number of crossing points.  Similarly, for forecasts with shared discrimination ability, Murphy curves assess calibration only.  

Following the pioneering work of \citet{Murphy1973}, researchers have sought decompositions of mean scores into intuitively appealing components that reflect calibration and discrimination, respectively.  We utilize the CORP decomposition of \citet{Dimitriadis2021}, which decomposes a mean score 
\begin{align*}
\SX =  \MCB - \DSC + \UNC
\end{align*}
into readily interpretable components that represent miscalibration (\MCB), discrimination (\DSC), and uncertainty (\UNC), respectively.  In contrast to earlier approaches, CORP reliability diagrams and CORP score components do not depend on user choices or tuning parameters, and they show appealing finite and large sample optimality properties.  The mean score $\SX$ equals a weighted area under the Murphy curve and serves as a summary measure of predictive performance, the {\MCB} component quantifies deviations of the CORP reliability diagram from the diagonal and can be used as a calibration metric, and the {\DSC} component serves as an appealing alternative to the widely used Area Under the ROC Curve (\AUC) measure of discrimination ability.  

If many competing forecasting methods are to be compared, the triptych graphics yield crowded displays.  With such settings in mind, we propose a simple alternative, namely, {\MCB}--{\DSC} plots, that show, for each competitor involved, the {\DSC} measure plotted against the {\MCB} component, augmented by parallel contour lines that indicate an equal mean score.  Due to their simplicity and the joint assessment of overall predictive ability, calibration, and discrimination, {\MCB}--{\DSC} plots visualize strengths and weaknesses of forecasting methods and facilitate the identification of methods of interest that can be analyzed in more detail via the triptych graphics.  In Figure \ref{fig:MCB_DSC} we show Brier score {\MCB}--{\DSC} plots for probability forecasts from solar flare and social science forecast contests. 

\begin{figure}[t]
\centering
\includegraphics[width=\linewidth]{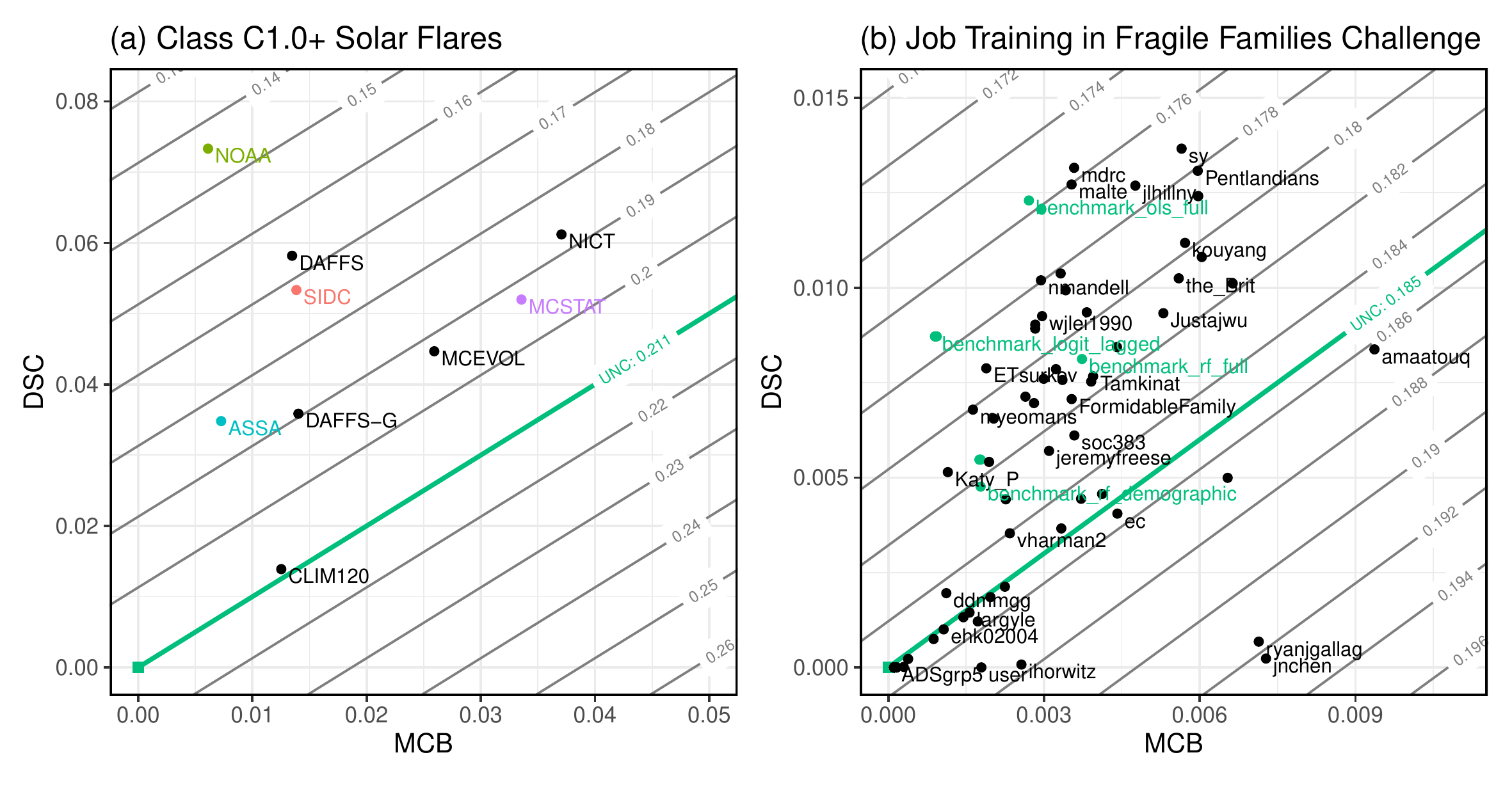}
\caption{Brier score {\MCB}--{\DSC} plots for competitors in forecast contests for (a) class C1.0+ solar flares \citep{Leka2019a}, and (b) job training in the Fragile Families Challenge \citep{Salganik2020b}.  The colors in panel (a) align with Figure \ref{fig:C1_triptych}; in panel (b) benchmark forecasts are represented in green.  The green square at the origin represents the ex post best constant forecast, that is, the unconditional event frequency, and the thick green line separates forecasts that are better (above the line) and that are worse (below the line) than this baseline.  Details of the data examples from astrophysics and social science are discussed in Sections \ref{sec:solar} and \ref{sec:FFC}, respectively.}  
\label{fig:MCB_DSC}
\end{figure}

While there is a rich literature on the evaluation of probabilistic classifiers and associated graphical displays, as reviewed by \citet{Murphy1992}, \citet{Prati2011}, \citet{Richardson2012}, \citet{Alba2017}, \citet{Filho2021}, and \citet{Xenopoulos2023}, and variants of the joint triptych display feature in the extant literature, as in Figure 1 of \citet{Flach2017} and graphics in \citet{TaillardatMestre2020}, the original contributions of our work include the presentation of the triptych as an argumentatively complete set of displays, connections to the CORP approach of \citet{Dimitriadis2021}, the development of {\MCB}--{\DSC} plots, and novel theoretical results proved in Appendix \ref{app:proofs}. 
  
The remainder of the article is organized as follows.  Sections \ref{sec:scores}, \ref{sec:reliability}, and \ref{sec:ROC} concern the individual triptych displays, by discussing proper scoring rules and Murphy curves, CORP reliability diagrams and score decompositions, and ROC curves, respectively.   Section \ref{sec:together} argues for the simultaneous use of the triptych displays, studies the connections between the individuals displays, and discusses {\MCB}--{\DSC} plots.  In particular, we show that for two competing forecasts that are both calibrated, Murphy curves and ROC curves yield congruent insights, as they share the same number of crossing points.  In Section \ref{sec:empirical} we apply the proposed methods in case studies from astrophysics, economics, and social science.  The paper closes in Section \ref{sec:discussion}, where we discuss open problems and directions for future research.  Software for the implementation of the proposed tools is available in \textsf{R} \citep{R, replication_triptych}.

\section{Scoring rules assess overall predictive performance}  \label{sec:scores}

Comparative assessments of overall forecast quality rely on proper scoring rules that encourage honest and careful forecasting \citep{Brier1950, Gneiting2007a}.  A scoring rule is a function $\myS(x,y)$ that assigns a numerical score or penalty based on the probability forecast $x \in [0,1]$ and the binary outcome $y \in \{ 0, 1\}$, where 1 stands for an event and 0 for a non-event.  Infinite penalties are permitted only if an outcome was declared to have probability zero and, thus, be impossible.  Throughout the paper we assume that scoring rules are negatively oriented, so that smaller scores are preferable.

\subsection{Proper and strictly proper scoring rules}  \label{sec:proper}

A scoring rule is proper if, given a Bernoulli random variable $Y$ with success probability $p$, 
\begin{align}  \label{eq:proper}  \textstyle
\myE \left[ \myS(p,Y) \right] \leq \myE \left[ \myS(x,Y) \right]
\end{align}
for all forecast values $x$.  It is strictly proper if, furthermore, equality in \eqref{eq:proper} implies that $x = p$, so that the true success probability is the unique minimizer of the expected score.  The key benefit of propriety is the implicit enforcement of honest and careful forecasts: If a forecaster believes that an event has success probability $p$, then $p$ is her best forecast in terms of the expected score or penalty.  In practice, for a given record
\begin{align}  \label{eq:data} 
(x_1, y_1), \ldots, (x_n, y_n)
\end{align} 
of probability forecasts $x_1, \ldots, x_n$ and associated binary outcomes $y_1, \ldots, y_n$, the mean score 
\begin{align}  \label{eq:SX} 
\SX = \frac{1}{n} \sum_{i=1}^n \myS(x_i,y_i)     
\end{align}
is used to rank competing forecasts. Evidently, the expression on the right-hand side of equation \eqref{eq:SX} corresponds to the expectation $\myE \left[ \myS(X, Y) \right]$ when the tuple $(X, Y)$ of random quantities follows the joint empirical distribution of the record \eqref{eq:data}.  

The most popular examples of strictly proper scoring rules are the Brier score (\BS) and the Logarithmic score (\LogS), defined by 
\begin{align}
\BS(x, y) =  (x-y)^2 \quad \text{and} \quad \LogS(x,y) = - y \log(x) - (1-y) \log(1-x)
\end{align}
for $x \in [0,1]$ and $y \in \{ 0, 1 \}$.  The zero--one loss or score
\begin{align}  \label{eq:MR}  \textstyle
\myS_{\frac{1}{2}}(x, y) = \mathds{1} \! \left( x > \frac{1}{2}, y = 0 \right) 
                         + \mathds{1} \! \left( x < \frac{1}{2}, y = 1 \right) 
                         + \frac{1}{2} \, \mathds{1} \! \left( x = \frac{1}{2} \right)
\end{align}
is a prominent example of a scoring rule that is proper, but not strictly proper.  When averaged in the form of \eqref{eq:SX}, it yields the widely reported misclassification rate.  The zero--one loss arises as the special case $\theta = \frac{1}{2}$ of the general elementary scoring rule 
\begin{align}  \label{eq:elementary}  \textstyle
\myS_\theta(x,y) = 2 \theta \, \mathds{1}(x > \theta, y = 0) 
                 + 2 (1-\theta) \, \mathds{1}(x < \theta, y = 1)
                 + 2 \theta(1-\theta) \, \mathds{1}(x = \theta)
\end{align}
with decision threshold or cost--loss parameter $\theta \in (0,1)$, which is proper, but not strictly proper, as it only takes into account whether a predicted probability is smaller or larger than $\theta$, so cannot distinguish between forecasts that are on the same side of $\theta$.  From an economic perspective, $\myS_\theta$ specifies the loss of a rational decision maker when the ratio of the monetary cost of a false alarm versus the cost of a missed event equals $\theta/(1-\theta)$; see \citet{Ehm2016} and references therein.\footnote{The economic interpretation applies to the left-continuous version of $\myS_\theta$ in eq.~(14) of \citet{Ehm2016}.  Here we use the symmetric version in \eqref{eq:elementary}, which assigns a fixed penalty of $2 \theta(1-\theta)$ when $x = \theta$, independently of the binary outcome $y \in \{ 0, 1 \}$.  Both versions are proper, but not strictly proper.}  In turn, $\myS_\theta$ can be identified with the special case $t = c = \theta$ of the general, cost-weighted misclassification loss at decision threshold $t$ and cost proportion $c$, as studied in the machine learning literature \citep{Hand2009, HernandezOrallo2011, HernandezOrallo2012, HernandezOrallo2013}. 

\subsection{Representations of proper scoring rules}  \label{sec:mixture}

The special role of the elementary scoring functions $\myS_\theta$ from \eqref{eq:elementary} is highlighted in a mixture representation studied by \citet{Schervish1989}.  Subject to technical conditions that are immaterial in practice, every proper scoring rule admits a representation of the form 
\begin{align}  \label{eq:Schervish}
\myS(x,y) = \int_0^1 \myS_\theta(x,y) \: \mathrm{d}H(\theta)
\end{align}
for forecast values $x \in [0,1]$ and outcomes $y \in \{ 0, 1 \}$, where $H$ is a measure that assigns non-negative weight to cost--loss parameters $\theta \in (0,1)$.  If the assigned weight is positive almost everywhere, then the corresponding score is strictly proper.  The elementary score $\myS_\eta$ arises when $H$ is a point measure that assigns mass one to $\eta \in (0,1)$ and no mass elsewhere, the Brier score emerges when the mixing measure $H$ is uniform, and the logarithmic score arises when $H$ has density proportional to $(\theta(1-\theta))^{-1}$.  Hence, the logarithmic score assigns infinite mass to the integrand in \eqref{eq:Schervish} at the very boundaries of the unit interval.  It discourages predictions with forecast probabilities at or near 0 or 1, and may render a single (and, hence, a mean) score infinite, as for the ASSA forecast from Table \ref{tab:C1}.  As Nobel laureate \citet[p.~51]{Selten1998} argued,  
\begin{quote} 
\footnotesize
``The use of the logarithmic scoring rule implies the value judgment that small differences between small probabilities should be taken very seriously and that wrongly describing something extremely improbable as having zero probability is an unforgivable sin.  The author thinks that this value judgment is unacceptable.''
\end{quote}
The mixture representation \eqref{eq:Schervish} also is a powerful tool for the construction of proper scoring rules.  For instance, \citet{Buja2005} introduce the flexible Beta family that arises when the mixing measure $H$ has density proportional to a Beta density.  The members of the Beta family include the Brier score, the logarithmic score, and the H-measure of \citet{Hand2009}. 

An alternative, essentially equivalent characterization is due to \citet{Savage1971}, who showed that, subject to technical conditions, any proper scoring rule allows a representation of the form 
\begin{align}  \label{eq:Savage}
\myS(x,y) = \phi(y) - \phi(x) - \phi'(x) (y-x),
\end{align}
where the function $\phi$ is convex with subgradient $\phi'$.  A subgradient is a generalized version of the classical derivative, and whenever the latter exists the subgradient equals the derivative.  In relation to the mixture representation from \eqref{eq:Schervish} it holds that $\mathrm{d}H(\theta) = \mathrm{d}\phi'(\theta) = \phi''(\theta) \, \mathrm{d}\theta$, with slight technical adaptations when $\phi$ is convex, but not strictly convex \citep{Gneiting2007a}.  Under the Savage representation \eqref{eq:Savage} the Brier score arises when $\phi(t) = t^2$, the logarithmic score emerges when $\phi(t) = t \log(t) + (1-t) \log(1-t)$, and the elementary scoring rule $\myS_\theta$ arises under the convex, but not strictly convex, function $\phi(t) = 2 \max(\theta t, (1-\theta) (1-t))$. 

In practice, it is not uncommon that different proper scoring rules yield distinct forecast rankings.  For example, in Table \ref{tab:C1} the Brier score and the logarithmic score disagree in the ranking of the ASSA and MCSTAT forecasts.  As there is no obvious reason for a specific strictly proper scoring rule to be preferred over any other, a natural question is which one to choose \citep{Merkle2013}.  

\subsection{Murphy curves and Murphy diagrams}  \label{sec:Murphy}

The mixture representation \eqref{eq:Schervish} allows for a compelling resolution of the challenge for guidance in the choice of proper scoring rules.  As the representation shows, any strictly proper scoring rule arises as a mixture over the family of the elementary scoring functions $\myS_\theta$ from \eqref{eq:elementary}.  Thus, it suffices to consider the family of mean elementary scores, 
\begin{align}  \label{eq:SXtheta} 
\SXtheta = \frac{1}{n}\sum_{i=1}^n \myS_\theta(x_i,y_i),
\end{align}
where $\theta \in (0,1)$.  \citet{Ehm2016} proposed plots of the Murphy curve, that is, the graph of $\SXtheta$ as a function of the decision threshold or cost--loss parameter $\theta \in (0,1)$, which allows users to assess forecast performance with respect to all scoring rules simultaneously.  In particular, the height of a Murphy curve at $\theta = \frac{1}{2}$ equals the misclassification rate, but note that sole focus on the misclassification rate as a general measure of predictive performance is problematic, because any single $\myS_\theta$ represents a particular economic scenario as mentioned in Section \ref{sec:proper} and fails to be strictly proper.  For a general measure, it is better to look at the area under a Murphy curve which equals the mean Brier score.

\citet{HernandezOrallo2011} had proposed the same tool under the name of Brier curve, and related displays had been studied by \citet{Murphy1992} and \citet{Drummond2006}, among other authors.  A Murphy diagram then is a plot of Murphy curves for competing forecasts.  If a forecast exhibits a lower mean score than another under all $\myS_\theta$, then it dominates the competitor in terms of any proper scoring rule $\myS$, as lower values under $\myS_\theta$ carry over to $\myS$ through integration in the mixture representation \eqref{eq:Schervish}.

\begin{figure}[t]
\centering
\includegraphics[width=\linewidth]{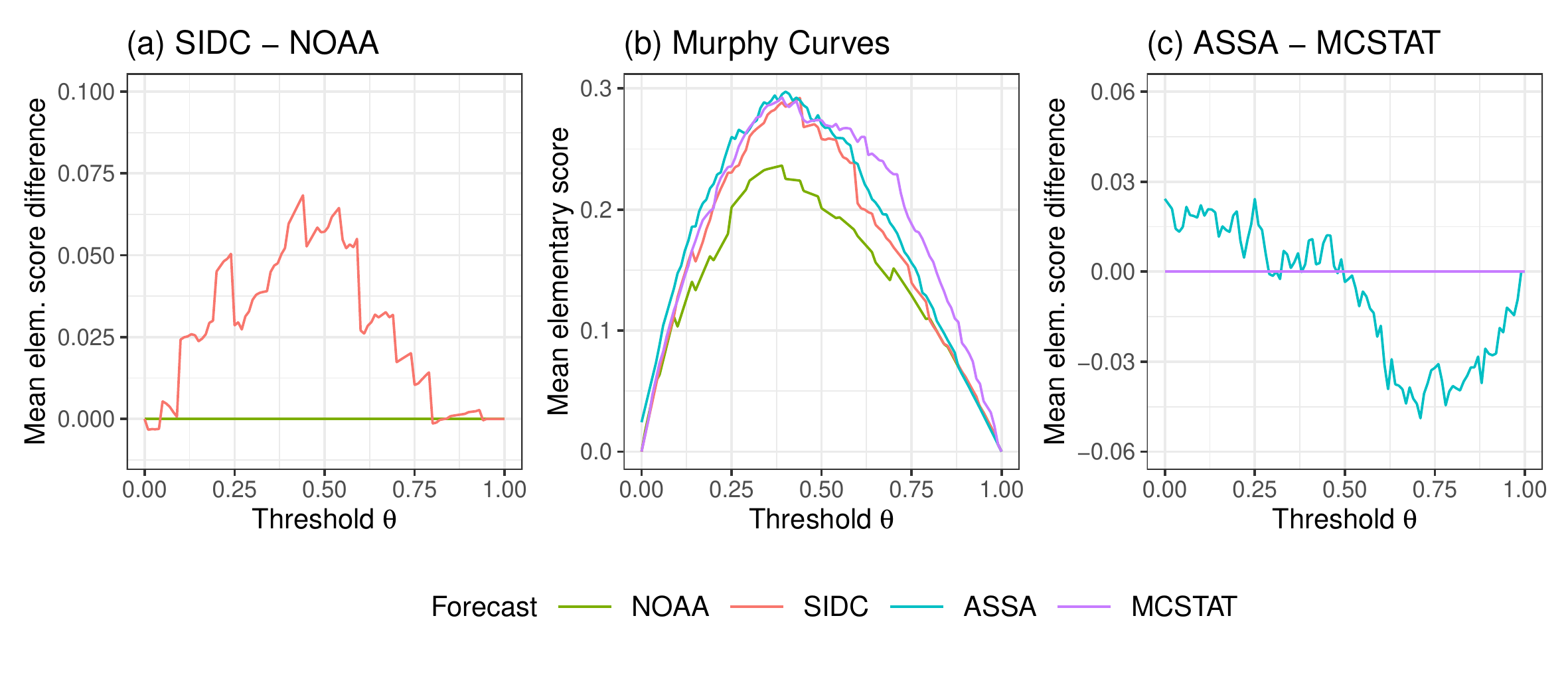}
\vspace{-9mm}
\caption{Murphy curves for the probability forecasts of class C1.0+ solar flares from Table \ref{tab:C1}.  Panel (b) shows the mean elementary score from \eqref{eq:SXtheta}; panels (a) and (c) show the difference between the mean elementary score for SIDC and NOAA, and for ASSA and MCSTAT, respectively.}  
\label{fig:C1_Murphy}
\end{figure}

In Figure \ref{fig:C1_Murphy} we take a closer look at the Murphy diagram for the class C1.0+ solar flare forecasts.  The panel at left compares the leading contenders from Table \ref{tab:C1}, the NOAA and the SIDC forecasts, and we note that NOAA has smaller $\SXtheta$ from \eqref{eq:SXtheta} nearly throughout, so that integration over $\SXtheta$ with respect to typically used mixing measures yields smaller mean scores.  The panel at right compares the MCSTAT and ASSA forecasts, and we see that, while for decision thresholds $\theta$ up to about 0.30 the former has lower $\SXtheta$, the situation is reversed for $\theta$ greater than 0.50.  Hence, it depends on the mixing measure $H$ in \eqref{eq:Schervish} whether a scoring rule prefers the MCSTAT or the ASSA forecast. 

To summarize, Murphy diagrams assess overall predictive performance, with the evaluation being complete in terms of proper scoring rules and economic utility.  Still, a more detailed assessment of the merits and deficiencies of competing forecasts is often desirable.  For example, a forecast might be deficient by systematically over- or underpredicting, or by an inability to distinguish between instances of higher and lower true CEPs.  In a nutshell, these two types of deficiencies correspond to lack of  calibration and lack of discrimination, respectively.  While Murphy diagrams rank forecasts with such deficiencies lower, they are not capable of diagnosing the form and extent of these issues.  Reliability diagrams and ROC curves, to which we turn in the sequel, serve this purpose. 

\section{Reliability diagrams assess calibration}  \label{sec:reliability}

A crucial, desirable property of a probabilistic classifier is that, when looking back at a collection of forecasts and associated binary outcomes, whenever the forecast value $x$ was issued, the outcome ought to occur in about $100 \cdot x$ percent of the respective instances.  To formalize this property, it is useful to think of the forecast and the outcome as random variables $X$ and $Y$, respectively, with joint distribution $\myQ$.  Then the probability forecast $X$ is calibrated if the conditional event probability, 
\begin{align}  \label{eq:CEP} 
\CEP(x) = \myQ \left( Y=1 \mid X=x \right) = \myE \left[ Y \mid X=x \right] \! ,
\end{align}
agrees with the forecast value $x$ for all relevant $x \in [0,1]$.  As \citet[Theorem 2.11]{Gneiting2013} have shown, the condition in \eqref{eq:CEP} serves as as unified notion of calibration for binary outcomes.

\subsection{Reliability curves and reliability diagrams}  \label{sec:reliabilitydiagrams}

Calibration is typically assessed graphically via reliability diagrams \citep{Murphy1977a, Murphy1992, Brocker2007} that plot an estimated version of the conditional event probability $\CEP(x)$ against the forecast value $x$, with deviations from the diagonal suggesting lack of calibration.  Classical approaches to estimating $\CEP(x)$ rely on binning and counting and have been hampered by ad hoc implementation decisions, instability under unavoidable choices regarding binning, and inefficiency \citep{Dimitriadis2021, ArrietaIbarra2022, Roelofs2022}.  To resolve these issues, \citet{Dimitriadis2021} introduced the CORP (Consistent, Optimally binned, Reproducible, and Pool-Adjacent-Violators (PAV) algorithm based) reliability diagram that plots an estimate of $\CEP(x)$ obtained through nonparametric isotonic regression, subject to the regularizing constraint of isotonicity in $x$, as implemented via the Pool-Adjacent-Violators (PAV) algorithm \citep{Ayer1955, deLeeuw2009, Jordan2022}.  

While isotonic regression and the PAV algorithm have been well known as tools for re-calibration \citep{Zadrozny2002}, their usage in the construction of reliability diagrams is a recent development.  For a given record of the form \eqref{eq:data} suppose without loss of generalization that $x_1 \leq \cdots \leq x_n$, and let 
\begin{align}  \label{eq:xhat}
\hat{x}_1 \leq \cdots \leq \hat{x}_n
\end{align}
denote the re-calibrated values generated by the PAV algorithm.  The CORP reliability diagram shows the piecewise linear curve that connects the points $(x_1, \hat{x}_1), \ldots, (x_n, \hat{x}_n)$.  If the original forecast is calibrated, then $x_1 = \hat{x}_1, \ldots, x_n = \hat{x}_n$, and the reliability curve lies on the diagonal.  Otherwise, systematic deviations from the diagonal suggest a lack of calibration.  

In contrast to classical approaches that estimate a reliability curve by assigning forecast values to bins, and counting events per bin, which mandates user choices, the CORP approach does not require any tuning parameters or user intervention, benefits from the regularizing constraint of isotonicity, and has appealing finite sample optimality and asymptotic consistency properties \citep{Dimitriadis2021}.  If desired, CORP reliability curves allow for an interpretation in terms of binning and counting, by identifying any horizontal segment with a bin, and interpreting the CEP as the corresponding empirical event frequency.  Generally, we supplement the reliability curve with a histogram that depicts the unconditional distribution of the forecast values.  

\begin{figure}[t]
\centering
\includegraphics[width=\linewidth]{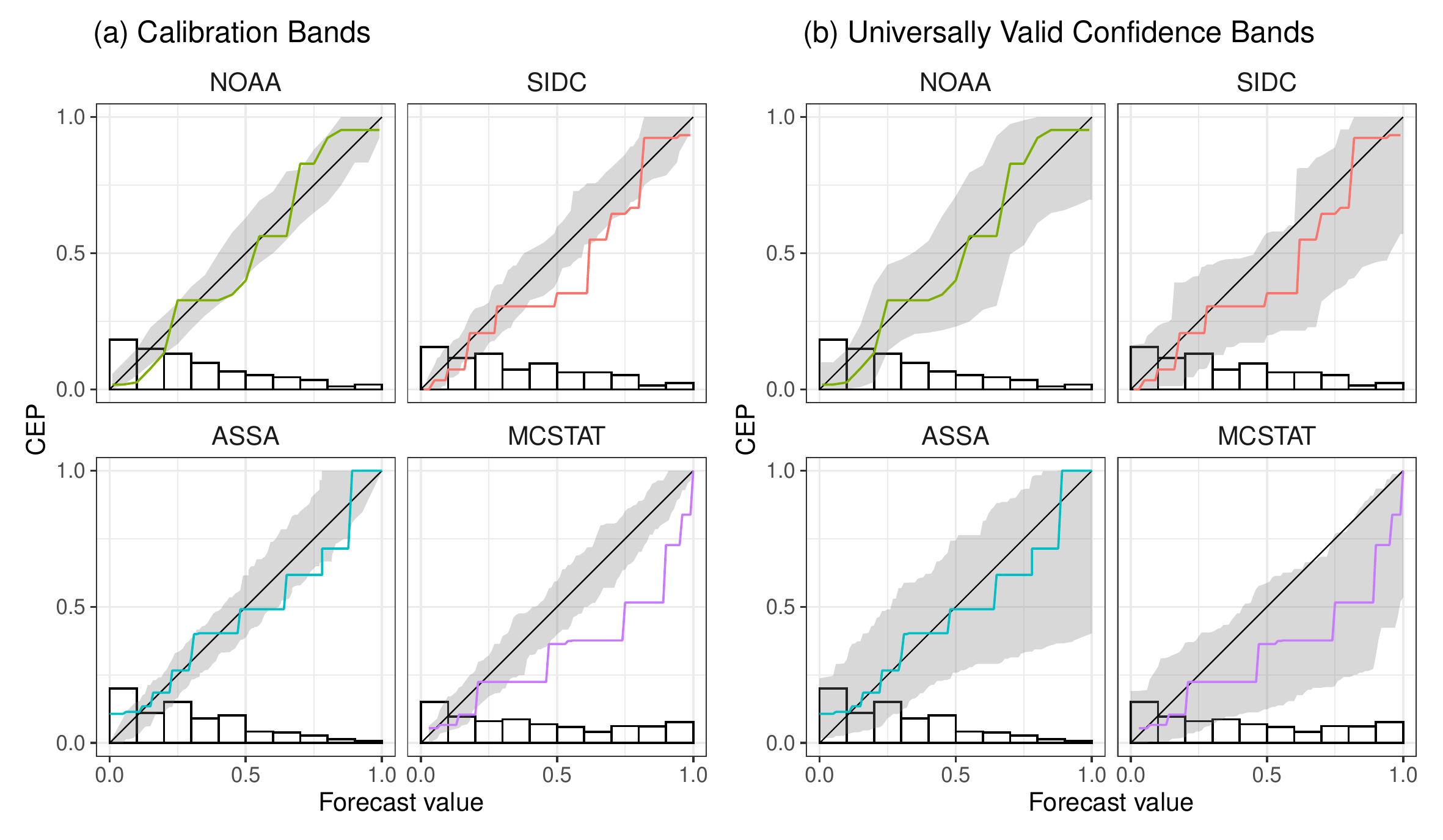}
\vspace{-7mm}
\caption{CORP reliability diagrams for the probability forecasts of class C1.0+ solar flares from Table \ref{tab:C1}, along with (a) consistency bands around the diagonal, and (b) confidence bands, both at the 90 percent level.}  \label{fig:C1_reliability}
\end{figure}	

In Figure \ref{fig:C1_reliability}, we return to the solar flare forecasts from Table \ref{tab:C1} and Figure \ref{fig:C1_triptych}.  At left the CORP reliability diagrams are supplemented by consistency bands.  Segments of a reliability curve that lie outside the consistency band, which represents 90 percent of the reliability curves that arise under the assumption of a calibrated forecast, suggest that the lack of calibration should not be attributed to estimation noise alone.  At right we show universally valid confidence bands \citep{Dimitriadis2022}.  The NOAA and ASSA forecasts are well calibrated, with reliability curves mostly within the consistency bands, and the diagonal well within the confidence bands.  In contrast, the SIDC and MCSTAT forecasts shows substantial underprediction.  

While consistency bands and confidence bands have complementary tasks \citep{Dimitriadis2021}, we use the former as a default, as they address the natural question whether (or not) observed differences between a CORP reliability curve and the diagonal can reasonably be attributed to chance alone, despite the forecast probabilities being perfectly calibrated. 

\subsection{Empirical score decomposition: Miscalibration (\MCB), discrimination (\DSC), and uncertainty (\UNC) components}  \label{sec:decomposition}

For several decades, researchers have sought decompositions of the mean score $\SX$ from \eqref{eq:SX} into nonnegative components that allow for persuasive interpretation \citep{Murphy1973, Ferro2012}.  Typically, a decomposition involves a reliability or miscalibration (\MCB) term that indicates how much the predicted probabilities differ from the conditional event frequencies, a resolution or discrimination (\DSC) term that measures a forecast's ability to distinguish between events and non-events, and an uncertainty (\UNC) component that quantifies the inherent difficulty of the prediction problem, but does not depend on the forecast under consideration.  While extant approaches lack stability under mandatory user decisions, particularly about binning, and may fail to provide an exact decomposition, or might yield components that fail to be nonnegative, the CORP approach yields a new type of decomposition that resolves these issues. 

As before, for a given record \eqref{eq:data} suppose without loss of generality that $x_1 \leq \cdots \leq x_n$, and let $\hat{x}_1 \leq \cdots \leq \hat{x}_n$ denote the PAV re-calibrated values from \eqref{eq:xhat}, as plotted in the CORP reliability diagram.  Furthermore, let $r = \bar{y} = \frac{1}{n} \sum_{i=1}^n y_i$ be the realized unconditional event frequency.  With $\myS$ being any proper scoring rule, let
\begin{equation}  \label{eq:SCR} 
\SC = \frac{1}{n} \sum_{i=1}^n \myS(\hat{x}_i,y_i) \quad \textrm{and} \quad \SR = \frac{1}{n} \sum_{i=1}^n \myS(r,y_i)
\end{equation} 
denote the mean score for the (re)Calibrated probabilities, and for the constant Reference forecast $r$, respectively.  Then the mean score $\SX$ from \eqref{eq:SX} decomposes as
\begin{equation}  \label{eq:decomposition} 
\SX = \underbrace{\left( \hspace{0.3mm} \SX - \SC \right)}_\MCB 
	- \underbrace{\left( \hspace{0.3mm} \SR - \SC \right)}_\DSC 
	+ \underbrace{\SR}_{\UNC}.
\end{equation} 
The miscalibration term $\MCB = \SX - \SC$ equals the difference in the mean score of the original versus the (re)calibrated forecast.  It expresses deviations of the CORP reliability curve from the diagonal in terms of the score under consideration.  The discrimination component $\DSC = \SR - \SC$ quantifies how much the (re)calibrated forecast improves upon the reference score $\SR$ that is based on a calibrated, but constant forecast, and we note that, by construction, {\DSC} is invariant under strictly increasing transformations of the forecast values.  While small values of {\MCB} are preferable, so are large values of the {\DSC} component.  The uncertainty term $\UNC = \SR$ is independent of the forecast at hand and provides a natural benchmark, as it equals the score of the (ex post) best constant forecast.  In contrast to earlier types of decomposition, the CORP decomposition from \eqref{eq:SCR} and \eqref{eq:decomposition} is exact and guarantees that $\MCB \geq 0$ with equality if the original forecast is calibrated, and $\DSC \geq 0$ with equality if the (re)calibrated forecast is constant \citep[Theorem 1]{Dimitriadis2021}.  

\begin{table}[t]
\centering
\footnotesize
\caption{CORP decomposition of the mean score in \eqref{eq:SX} for the probability forecasts of class C1.0+ solar flares from Table \ref{tab:C1}, under the Brier score ($\UNC = 0.211$), the logarithmic score ($\UNC = 0.614$), and misclassification rate ($\UNC = 0.303$).}  \label{tab:C1_decom}
\begin{tabular}{l r ccc r ccc r ccc}
\toprule
Forecast && \multicolumn{3}{c}{Brier Score} && \multicolumn{3}{c}{Logarithmic Score} && \multicolumn{3}{c}{Misclassification Rate} \\
\cmidrule{3-5} \cmidrule{7-9} \cmidrule{11-13} 
         && $\SX$ & {\MCB} & {\DSC} && $\SX$    & {\MCB}   & {\DSC} && $\SX$ & {\MCB} & {\DSC} \\ 
\midrule
NOAA     && 0.144 & 0.006  & 0.073  && 0.449    & 0.027    & 0.191  && 0.205 & 0.004  & 0.102 \\
SIDC     && 0.172 & 0.014  & 0.053  && 0.515    & 0.036    & 0.135  && 0.263 & 0.038  & 0.078 \\
ASSA     && 0.184 & 0.007  & 0.035  && $\infty$ & $\infty$ & 0.085  && 0.273 & 0.006  & 0.036 \\
MCSTAT   && 0.193 & 0.034  & 0.052  && 0.587    & 0.101    & 0.128  && 0.275 & 0.042  & 0.071 \\
\bottomrule
\end{tabular}
\end{table}

The CORP decomposition applies under any proper scoring rule $\myS$.  When $\myS$ is the Brier score, it agrees with the classical decomposition of \citet{Murphy1973} in  the special case where the bins reduce to unique forecast values with an associated nondecreasing sequence of conditional event frequencies \citep[Theorem 2]{Dimitriadis2021}.  Under the misclassification rate that arises from \eqref{eq:SX} under the zero--one score in \eqref{eq:MR}, the components admit appealing interpretations in terms of the original, the (re)calibrated, and the constant reference forecast being on the same or distinct side(s) of $\frac{1}{2}$.  Table \ref{tab:C1_decom} shows the CORP decomposition of the mean Brier score, the mean logarithmic score and the misclassification rate for the solar flare forecasts from Table \ref{tab:C1}.  The {\MCB} components confirm the visual appearance of the CORP reliability diagrams in Figures \ref{fig:C1_triptych} and  \ref{fig:C1_reliability}.  The NOAA forecast exhibits the least and MCSTAT the most pronounced lack of calibration.  As discussed, the mean score $\SX$ and the {\MCB} component under the logarithmic score are infinite for the ASSA forecast.  We defer consideration of the {\DSC} components to Section \ref{sec:ROC}, where we focus attention on ROC curves. 

\subsection{Calibration metrics and the Brier score {\MCB} component}  \label{sec:MCB} 

Recently, there has been a surge of interest in calibration metrics in the machine learning literature.  The widely used metric of the expected or estimated calibration error \citep[ECE:][]{Naeini2015, Guo2017} depends on binning and counting and thus is subject to the aforementioned types of instabilities \citep{Dimitriadis2021, Roelofs2022} and biases \citep{Brocker2012, Ferro2012}.  In a recent review, \citet[p.~3]{ArrietaIbarra2022} summarize that
\begin{quote} 
\footnotesize
``the classical empirical calibration errors based on binning vary significantly based on the choice of bins. The choice of bins is fairly arbitrary and enables the analyst to fudge results (whether purposefully or unintentionally).''
\end{quote} 
To address these issues, \citet{Roelofs2022} use equal-mass bins and select the number of bins as large as possible while preserving isotonicity in the calibration curve.  \citet{ArrietaIbarra2022} and \citet{Brocker2022} recommend graphical displays, calibration metrics, and tests, that are based on cumulative differences between predicted and observed event probabilities.  While approaches of this type have appealing mathematical properties, cumulative quantities lack intuition and ease of interpretation.  Similar to the method proposed by \citet{Roelofs2022}, the CORP approach to reliability diagrams enforces isotonicity, but uses the PAV algorithm to select the number and the arrangement of the bins in fully automated, optimal ways \citep{Dimitriadis2021}.  The CORP decomposition from \eqref{eq:SCR} and \eqref{eq:decomposition} is based on the CORP reliability curve, and when $\myS$ is the Brier score it yields an {\MCB} component that reduces to a classical calibration metric under modest conditions \citep[Theorem 2]{Dimitriadis2021}.  For these reasons, we propose the use of the Brier score {\MCB} component as a calibration metric.  

\section{Receiver operating characteristic (ROC) curves visualize discrimination ability}  \label{sec:ROC}

While calibration is an important quality of probabilistic classifiers, a calibrated forecast is not necessarily powerful, as it may lack the ability to discriminate between events of low and high event probability.  ROC curves are key tools in the assessment of this ability \citep{Egan1975, Swets1973, Fawcett2006}.  In a nutshell, ROC curves visualize potential predictive power, detached from considerations of calibration.  

\subsection{ROC curves}  \label{sec:ROCcurves}

To introduce ROC curves, suppose that we use the threshold $t$ to construct a hard classifier from the probability forecast $x$ in the usual way, by predicting an event ($y = 1$) if $x > t$, and predicting a non-event ($y = 0$) if $x \leq t$.  For a record of the form \eqref{eq:data}, the resulting False Alarm Rate (FAR) and Hit Rate (HR) are given by
\begin{align}
\HR(t) = \frac{\sum_{i=1}^n \mathds{1}(y_i = 1, x_i > t)}{\sum_{i=1}^n \mathds{1}(y_i = 1)} 
\quad \text{and} \quad 
\FAR(t) = \frac{\sum_{i=1}^n \mathds{1}(y_i = 0, x_i > t)}{\sum_{i=1}^n \mathds{1}(y_i = 0)},  
\end{align}	
respectively.  The ROC curve is the piecewise linear curve that connects the at most $n + 1$ unique points of the form $(\FAR(t), \HR(t))$ that arise as the threshold $t$ decreases.\footnote{Some researchers talk of ROC as relative operating characteristic \citep{Swets1973}, and the hit rate is also referred to as probability of detection, recall, sensitivity, or true positive rate.  The false alarm rate is also known as probability of false detection, fall-out, or false positive rate.  It equals one minus the specificity, selectivity, or true negative rate.  See \url{https://en.wikipedia.org/wiki/Precision_and_recall\#Definition_(classification_context)}, accessed 27 No\-vember 2022.  Moreover, some researchers define the ROC curve as a display that connects the points $(1-\HR(t), 1-\FAR(t))$ for $t \in [0,1]$, resulting in a curve that is mirrored at the anti-diagonal of the unit square, which maintains the interpretation that ROC curves at upper left are desirable \citep{HernandezOrallo2013}.}  Informally, the threshold $t$ parameterizes the ROC curve, with the points (0,0) and (1,1) corresponding to $t \geq 1$ and $t < 0$, respectively.  

ROC curves assess the discrimination ability of forecasts and can be interpreted diagnostically \citep{Marzban2004}.  If the empirical conditional distributions for data \eqref{eq:data} given an event ($y = 1$) and a non-event ($y = 0$) coincide, then the forecast is unable to distinguish between events and non-events, and its ROC curve lies on the diagonal.  The larger the separation between these conditional distributions, the higher the discriminatory power of the forecast, and the further to the upper left the ROC curve.  For a perfectly discriminating probabilistic classifier, there is a threshold value $t$ such that $y_i = 0$ if $x_i \leq t$ and $y_i = 1$ if $x_i > t$, and hence it exhibits an ideal ROC curve along the left and upper edges of the unit square.  By construction, ROC curves are invariant under strictly increasing transformations of the forecast values.  

\begin{figure}[tbp]
\centering
\includegraphics[width=\linewidth]{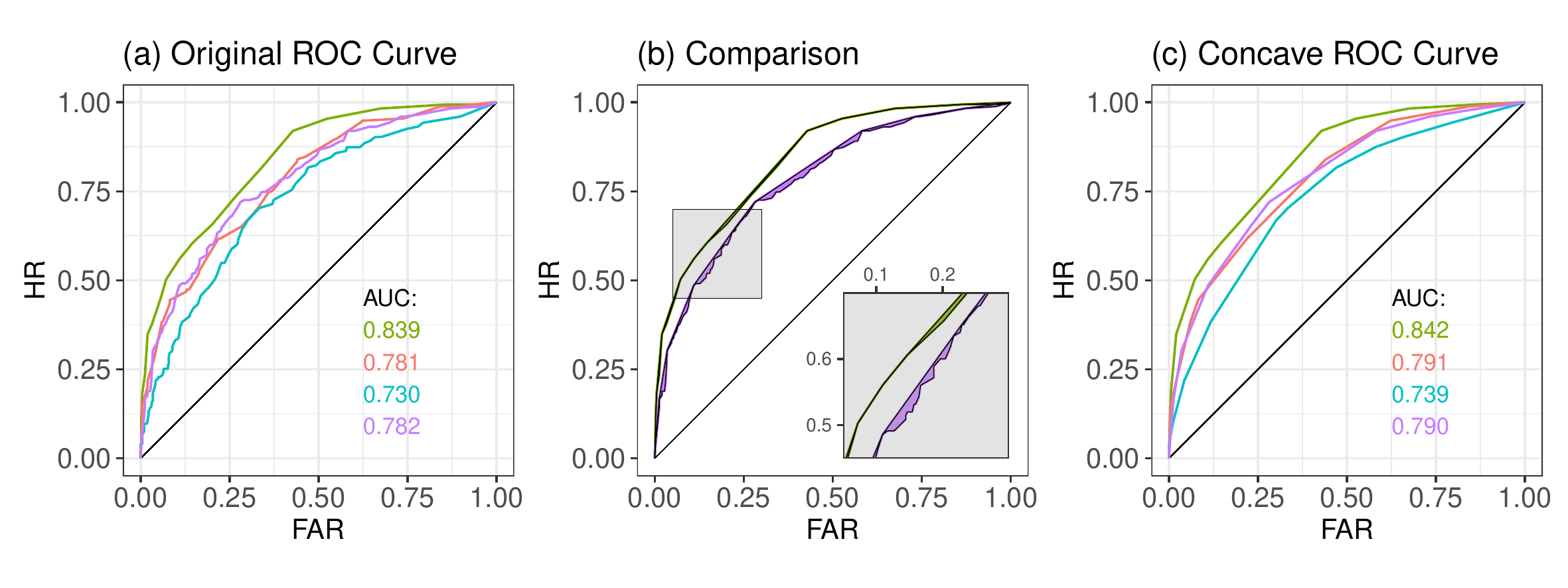}

\vspace{-2mm}

\includegraphics[width=0.56\linewidth]{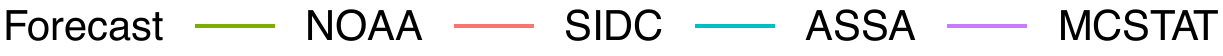}

\caption{ROC curves for the probability forecasts of class C1.0+ solar flares from Table \ref{tab:C1}.  Panel (a) shows the original ROC curves for the forecasts $x_1, \ldots, x_n$ from \eqref{eq:data}, panel (c) the concave ROC curves for the PAV re-calibrated forecasts $\hat{x}_1, \ldots, \hat{x}_n$ from \eqref{eq:xhat}.  In panel (b) both the original and the concave ROC curves are shown for the NOAA and SIDC forecasts.  The magnified details demonstrate that the original ROC curve morphs into its concave hull.}
\label{fig:C1_ROC}
\end{figure}

An often neglected but important consideration concerns the concavity of ROC curves.  In practice, the original ROC curves constructed from empirical data \eqref{eq:data} almost inevitably fail to be concave, as illustrated on the forecasts from Table \ref{tab:C1} in Figure \ref{fig:C1_ROC}.  This observation is explained by Theorems 3 and 4 of \citet{Gneiting2022a}, according to which a ROC curve is concave if, and only if, the conditional event probability is nondecreasing with the forecast value $x$, which for empirical data is hardly ever the case.  ROC curves assess discrimination ability --- that is, potential predictive ability --- and while potential predictive ability is ignorant of calibration, it can only be assessed under the assumption of larger forecast values implying higher event probabilities.  In this light, displays of nonconcave ROC curves have been harshly criticized, with researchers positing that they ``must be considered irrational'' and ``unethical when applied to medical decisions'' \citep{Pesce2010}. 

Fortunately, there is a straightforward remedy.  If one computes the ROC curve from the PAV transformed forecast values \eqref{eq:xhat} in lieu of the original forecasts from \eqref{eq:data}, the ROC curve morphs into its concave hull, that is, the smallest concave curve that lies to its upper left \citep{Fawcett2007}.  The corrected, concave version of the ROC curve serves exclusively to compare discrimination ability, as differences in calibration get eliminated through re-calibration.  In contrast, while original ROC curves focus on discrimination ability, conditional events frequencies that fail to be monotone generate confounding effects.  In general, the transformation from the original probabilities $x_1 \leq \cdots \leq x_n$ to the PAV transformed, re-calibrated probabilities $\hat{x}_1 \leq \cdots \leq \hat{x}_n$ is monotonic, but not strictly monotonic, so a change in the ROC curve does not contradict the aforementioned invariance under strictly increasing transformations.  

In summary, we strongly recommend the use of concave ROC curves, as computed from PAV transformed forecast values, in triptych graphics for empirical data.  In Figure \ref{fig:C1_ROC} the left-hand panel illustrates the original versions of the ROC curves for the solar flare forecasts, the right-hand panel displays the concave ROC curves, as in the triptych graphics in Figure \ref{fig:C1_triptych}, and the middle panel illustrates the transition from the original curve to the concave hull.  The NOAA forecast clearly discriminates the most and the ASSA forecast the least.  The MCSTAT and SIDC forecasts exhibit roughly equal discrimination ability, with ROC curves that are nested in between the curves for the NOAA and ASSA forecasts.

\subsection{The area under the curve (\AUC) measure and the Brier score discrimination (\DSC) component}  \label{sec:AUC}

Myriads of scientific papers have employed the Area Under the ROC Curve \citep[AUC:][]{Hanley1982, DeLong1988, Bradley1997, Marzban2004} measure to compare the predictive performance of probabilistic classifiers.  AUC admits an appealing interpretation as the probability of a value drawn at random from the empirical distribution of forecast values for an event being higher than a value drawn from the distribution for a non-event.  An {\AUC} value of 1 signifies perfect discrimination ability; a value of $\frac{1}{2}$ indicates no discrimination, corresponding to the trivial ROC curve on the diagonal.  A value smaller than $\frac{1}{2}$ implies that interchanging the predictions for 0 and 1 would improve forecast accuracy.  As the ROC curve is invariant under strictly increasing transformations, so is {\AUC}, and we note that {\AUC} exclusively concerns discrimination ability, while ignoring (mis)calibration.  

However, {\AUC} is not suitable as an overall performance metric.  To summarize arguments of \citet{Hand2009} in an informal way, and specializing them to calibrated forecasts, {\AUC} assumes the form in \eqref{eq:Schervish} with a mixing measure $H$ that depends on the forecast values in complex ways.  As argued by \citet{Hand2009}, such a dependence ``is absurd'', with \citet{Hand2022} recently adding that
\begin{quote}
\footnotesize
``The dependence of the {\AUC} on the classifier itself means that it can lead to seriously misleading results, so it is concerning that it continues to be widely used.''
\end{quote}
Both {\AUC} and the {\DSC} measure from the CORP score decomposition in \eqref{eq:decomposition} are measures of discrimination ability.  Under calibration, the {\MCB} component vanishes, and as the {\UNC} component is independent of the forecast, $\SX$ equals the {\DSC} component, except for an additive constant.  Furthermore, if $\myS$ is the Brier score then $\DSC$ reduces to a classical component, subject to conditions \citep[Theorem 2]{Dimitriadis2021}. 

\begin{figure}[t]
\centering
\includegraphics[width=\linewidth]{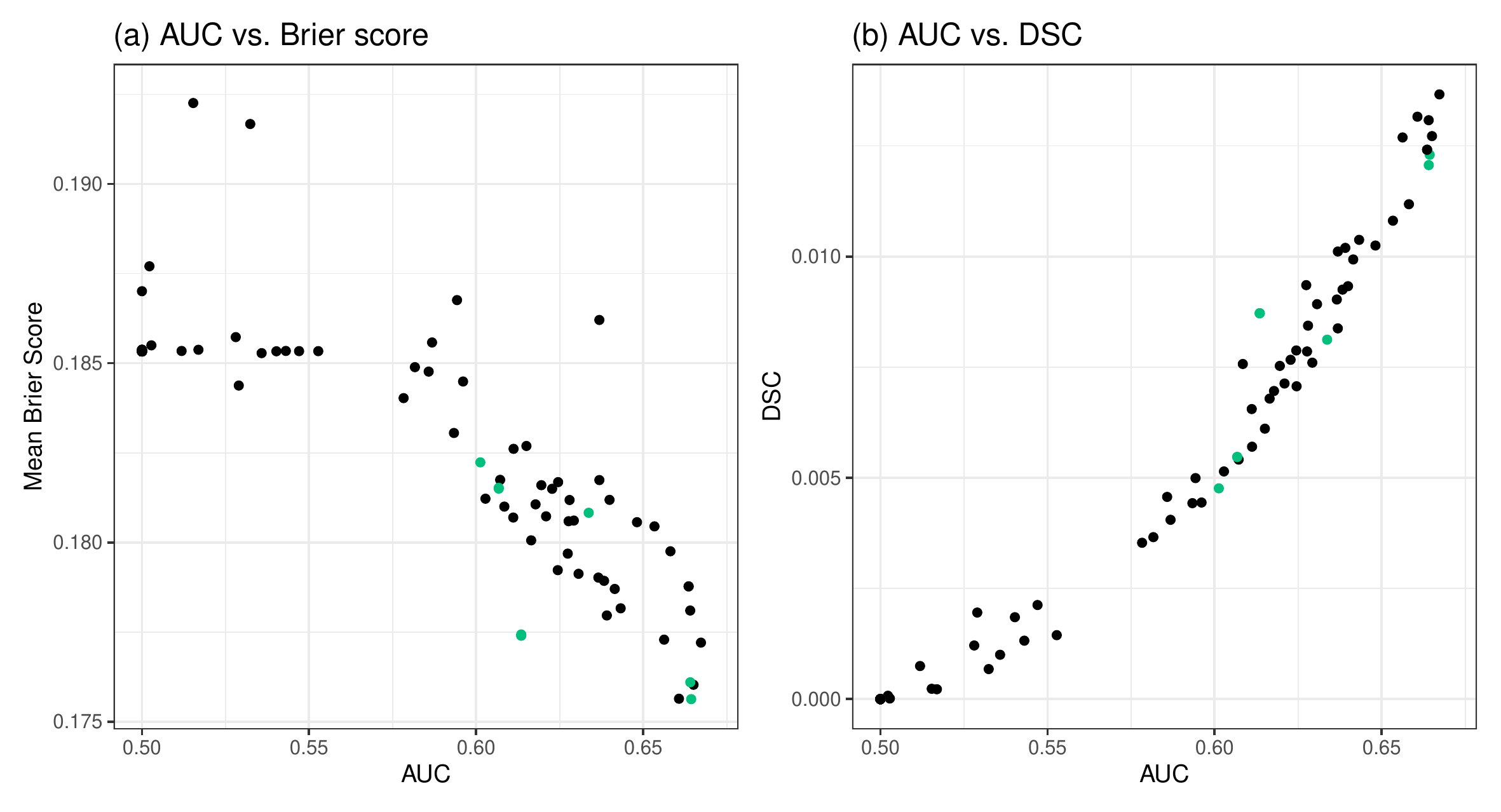}
\caption{Scatter plots of (a) {\AUC} vs.~mean Brier score, and (b) {\AUC} vs.~the Brier score {\DSC} component for the forecasts from panel (b) of Figure \ref{fig:MCB_DSC} with mean Brier score less than 0.195.  Benchmark forecasts are represented in green.  While {\AUC} and the Brier score {\DSC} component are positively oriented, the mean Brier score is negatively oriented.}
\label{fig:AUC_DSC}
\end{figure}

These relationships suggest the use of the Brier score {\DSC} component as an attractive alternative to {\AUC} if a measure of discrimination ability is sought.  For illustration, Figure \ref{fig:AUC_DSC} shows scatter plots of performance measures for the solar flare forecasts from Figure \ref{fig:MCB_DSC}.  In contrast to the association between {\AUC} and the mean Brier score itself, which is weak, the association between {\AUC} and the Brier score {\DSC} component is strong.  While both {\AUC} and {\DSC} are invariant under strictly increasing classifier transformations, $\DSC = \SR - \SC$ admits the appealing interpretation as (the negative of) the mean Brier score for the PAV (re)calibrated forecast, up to the constant $\SR$.  Evidently, this interpretation is retained when $\myS$ is the logarithmic score or any other proper scoring rule. 

\section{Putting it together: The triptych graphics and {\MCB}--{\DSC} plots}  \label{sec:together}

Considering the class C1.0+ solar flares example from the previous sections and Figures \ref{fig:C1_triptych}--\ref{fig:C1_ROC}, the NOAA forecast is superior in all facets.  However, rankings of the other forecasts from Table \ref{tab:C1} depend on the criteria used: The ASSA forecast is well calibrated, but exhibits poor discrimination ability.  The MCSTAT and SIDC forecasts show discrimination ability in between the NOAA and ASSA forecast, but lack calibration.  Murphy curves provide an overall assessment of predictive performance, covering both calibration and discrimination ability, and favor the NOAA forecast, followed by the SIDC forecast, whereas rankings of the MCSTAT and ASSA forecasts depend on the scoring rule used.  

These considerations illustrate that the evaluation of predictive performance is multi-faceted and highly complex, even in the basic setting of probabilistic classifiers for binary outcomes.  To address this challenge, we propose the use of a triptych of diagnostic graphics, consisting of Murphy curves, CORP reliability diagrams, and ROC curves (Figure \ref{fig:C1_triptych}).  The triptych displays present a wealth of information in disentangled form: The CORP reliability diagram diagnoses calibration (only), the ROC curve in its concave form assesses discrimination ability (only), and Murphy curves consider economic utility and overall predictive performance, encompassing both calibration and discrimination.   

\subsection{Theoretical guarantees}  \label{sec:theory}

In the triptych graphics, information about overall predictive ability in the Murphy curves is disentangled into facets of calibration, as displayed in CORP reliability diagrams, and facets of discrimination ability, as visualized by ROC curves.  We now summarize existing and new theoretical results that support this intuition and provide new insights about links between the displays, with technical details being available in Appendix \ref{app:proofs}.  The findings are illustrated in idealized settings, where we consider the joint distribution of a pair $(X,Y)$ of random variables, with $X$ representing the probability forecast and $Y$ the binary outcome.  The triptych graphics in these ideal settings derive from the population quantities, and can be interpreted as the triptych graphics that arises when a record of the form in \eqref{eq:data} is generated by ever larger samples from the joint distribution of $(X,Y)$.  As the populations involved show nondecreasing CEPs, original and re-calibrated probabilities coincide, and yield the same, concave ROC curve. 

We begin with a discussion of the role of calibration.  As noted, any forecast can be re-calibrated ex post, by applying the PAV algorithm that converts the original forecast probabilities $x_1 \leq \cdots \leq x_n$ from \eqref{eq:data} into the calibrated probabilities $\hat{x}_1 \leq \cdots \leq \hat{x}_n$ from \eqref{eq:xhat}.  The following stylized fact summarizes findings from \citet[Theorem 6.3]{Schervish1989} and \citet[Corollary 2]{Holzmann2014}.  

\begin{fact}  \label{fact:re-calibrated}
If a probability forecast fails to be calibrated, its re-calibrated version is superior in terms of Murphy curves.
\end{fact}

To illustrate Fact \ref{fact:re-calibrated} we consider the idealized Scenario A, where the binary outcome $Y$ has event probability $X_0$, which is uniformly distributed on the unit interval.  We compare to the probability forecast $X_1 = \frac{3}{8} + \frac{1}{4} X_0$, which is a strictly increasing transformation of $X_0$.  Part (a) of Figure \ref{fig:ABC} shows idealized triptych plots, where the reliability diagrams are augmented by density plots for the unconditional distribution of the forecast values.  While $X_0$ and $X_1$ have the same discrimination ability and identical ROC curves, $X_1$ fails to be calibrated, whereas $X_0$ is calibrated.  In fact, $X_0$ is the re-calibrated version of $X_1$, and thus is superior in terms of Murphy curves.   

\begin{figure}[p]
\centering
\begin{subfigure}{\linewidth}
\caption{Scenario A}  \label{fig:A}
\includegraphics[width=\linewidth]{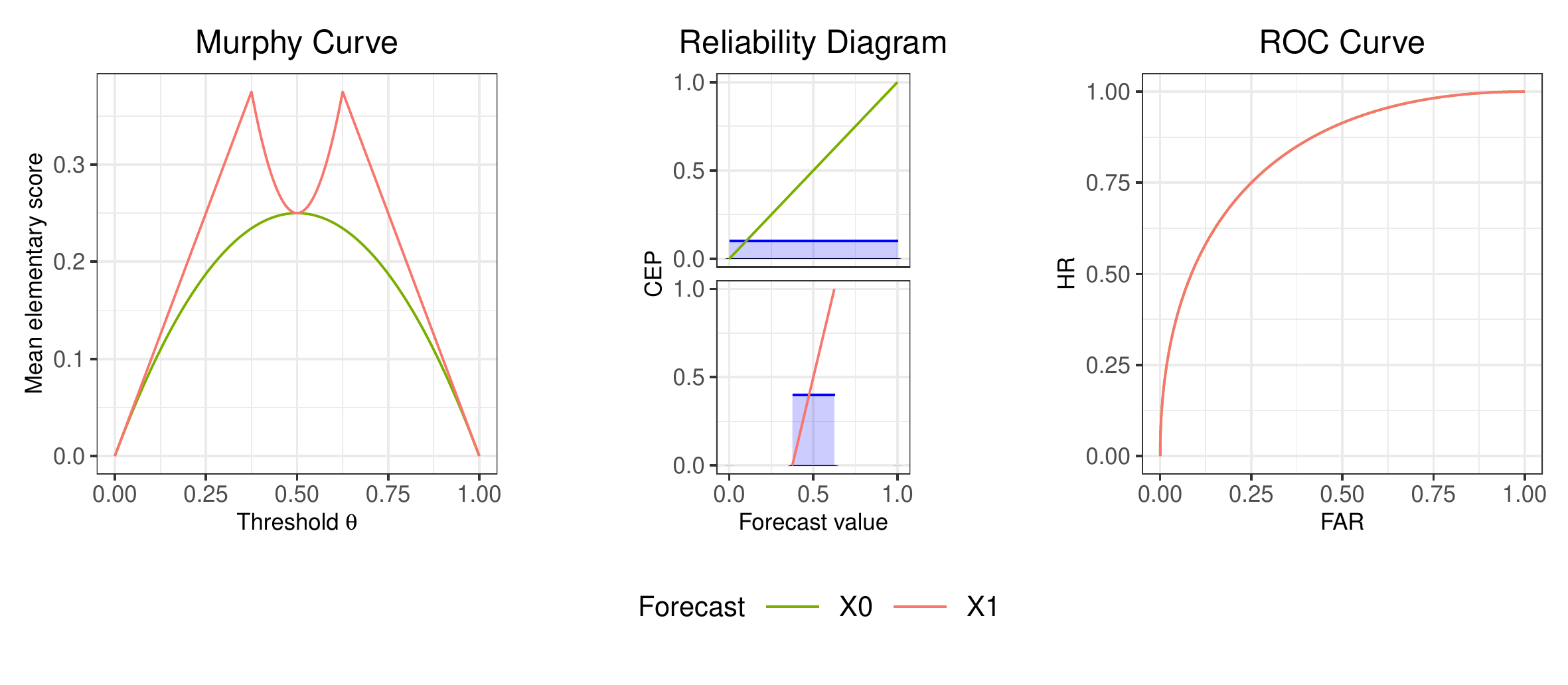}
\end{subfigure}
\\
\begin{subfigure}{\linewidth}
\caption{Scenario B}  \label{fig:B}		
\includegraphics[width=\linewidth]{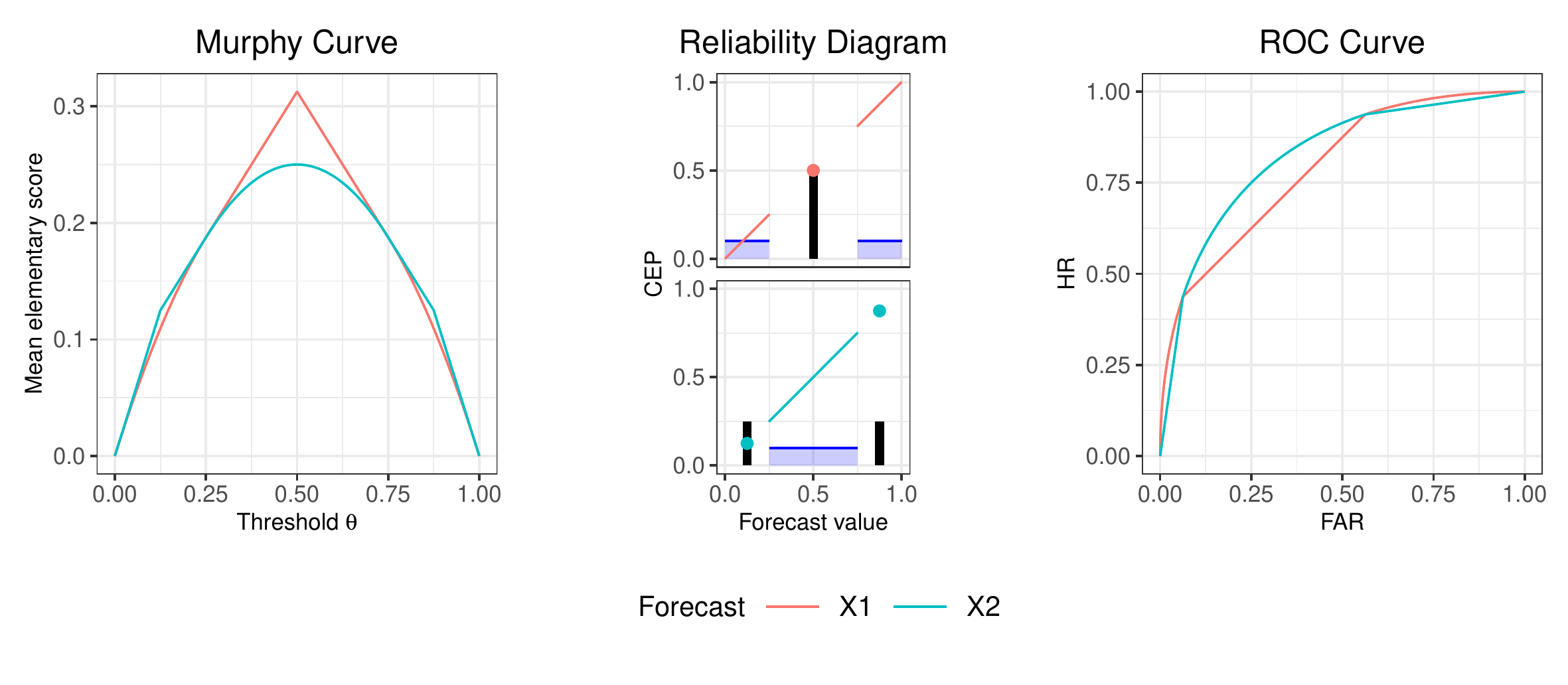}
\end{subfigure}
\\
\begin{subfigure}{\linewidth}
\caption{Scenario C}  \label{fig:C}
\includegraphics[width=\linewidth]{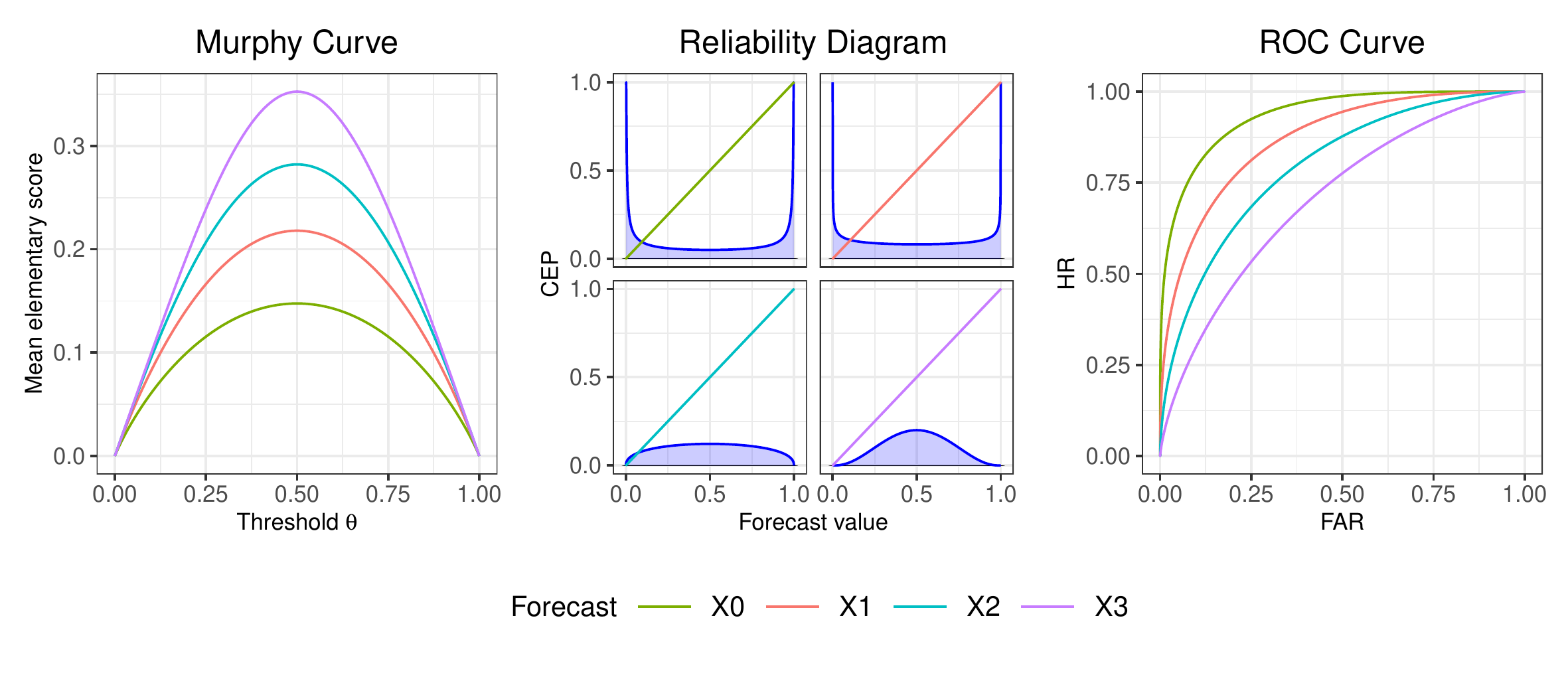}
\end{subfigure}
\vspace{-3mm}
\caption{Triptych displays in idealized Scenarios A, B, and C}  \label{fig:ABC}
\end{figure}

The next fact summarizes a crucial, novel finding.  We provide a rigorous version as Theorem \ref{thm:crossingpoints} in Appendix \ref{app:proofs}, which also contains its proof.

\begin{fact}  \label{fact:crossingpoints}
For two competing probability forecasts that are both calibrated, the number of crossing points of the ROC curves equals the number of crossing points of the Murphy curves.
\end{fact}
  
For illustration we tend to Scenario B, where the forecasts $X_1$ and $X_2$ are both calibrated.  As before, let $X_0$ be uniformly distributed, and let the outcome $Y$ have true event probability $X_0$.  We consider the probability forecasts  
\begin{align}  \label{eq:B} 
X_1 = \begin{cases}
X_0         & \text{if } X_0 < \frac{1}{4},  \\
\frac{1}{2} & \text{if } \frac{1}{4} \leq X_0 \leq \frac{3}{4}, \\
X_0         & \text{if } X_0 > \frac{3}{4}, 
\end{cases} 
\quad \text{and} \quad 
X_2 = \begin{cases}
\frac{1}{8} & \text{if } X_0 < \frac{1}{4}, \\
X_0         & \text{if } \frac{1}{4} \leq X_0 \leq \frac{3}{4}, \\
\frac{7}{8} & \text{if }  X_0 > \frac{3}{4},   
\end{cases} 
\end{align}
respectively.  The triptych plots in part (b) of Figure \ref{fig:ABC} illustrate that the ROC and Murphy curves share the same number, namely two, of crossing points. 

In particular, Fact \ref{fact:crossingpoints} implies that if two competing probability forecasts are both calibrated, then there is a superiority relation in terms of ROC curves if, and only if, there is a superiority relation in terms of Murphy curves.  The following fact sharpens this statement and relates to considerations of sharpness \citep{Gneiting2007}.  Informally, a probability forecast is sharper than another if its forecast values are closer to the most confident values of 0 and 1, respectively.  In Appendix \ref{app:proofs} we state and prove a rigorous version of the subsequent Fact \ref{fact:sharper} in Theorem \ref{thm:sharper}.  

\begin{fact}  \label{fact:sharper}
If two competing probability forecasts are both calibrated, and one of them is sharper than the other, then the sharper one is superior in terms of both ROC curves and Murphy curves, and vice versa. 
\end{fact}

For an illustration in terms of nested information sets, which imply Murphy dominance as proved by \citet[Corollary 4]{Holzmann2014} and \citet[Proposition 3.1]{Kruger2021}, we consider Scenario C.  Specifically, let the binary outcome $Y$ have true event probability $X_0 = \Phi( \, \sum_{i=1}^4 a_i \, )$, where $a_1, a_2, a_3$, and $a_4$, respectively, are independent standard normal variates, and $\Phi$ is the cumulative distribution function of the standard normal distribution.  We consider the probability forecasts 
\begin{align}  \label{eq:C}
X_j = \Phi \left( \frac{1}{(j+1)^{1/2}} \sum_{i=1}^{4-j} a_i \right) 
\end{align} 
for $j = 0, 1, 2$, and $3$.  So, there are four independent sources of information, represented by $a_1, a_2, a_3$, and $a_4$, and the forecast $X_j$ provides the correct specification of the event probability conditional on $4 - j$ sources being available.  Thus the information sets are nested, the forecasts are calibrated, and they exhibit an increase in sharpness as $j$ decreases.  The triptych graphs in part (c) of Figure \ref{fig:ABC}, which include density plots in the reliability diagrams, illustrate the increase in sharpness and the associated gain in terms of both ROC curves and Murphy curves.  Pairwise comparisons between the forecasts illustrate the relationships guaranteed by Fact \ref{fact:sharper}.

\subsection{Visualizing classifier performance for many competitors simultaneously: {\MCB}--{\DSC} plots}  \label{sec:MCB_DSC}

It is not uncommon that a multitude of competing forecasts are to be compared, with forecast contests being prime examples of such settings \citep{Leka2019a, Salganik2020b}.  The consideration of all competing forecasts in the triptych graphics then results in overcrowded displays.  However, the components of the CORP decomposition of a mean score $\SX$ from \eqref{eq:decomposition} can serve as numerical summaries.  For a succinct comparison that considers multiple facets of forecast performance we propose a simple tool, which we call an {\MCB}--{\DSC} plot, namely, a scatter plot of the miscalibration (\MCB) versus the discrimination (\DSC) component of the CORP decomposition, augmented with a set of parallel contour lines that according to \eqref{eq:decomposition} correspond to an equal mean score.  Importantly, the joint consideration of the {\MCB} and {\DSC} components enables a comparison in terms of the mean score $\SX$ as well, contrary to the joint use of the (traditional) Brier score reliability component and AUC in the extant literature, as exemplified in \citet[Figure 2]{Hewson2021}. 

{\MCB}--{\DSC} plots admit appealing interpretations that apply under any choice of the underlying proper scoring rule $\myS$, as summarized now.  
\begin{itemize}
\item
For any forecast method considered, the mean score $\SX$ and the associated {\MCB} and {\DSC} components from \eqref{eq:decomposition} can be read off immediately.  The {\UNC} component depends on the outcomes only, is shared by all methods considered, and equals the label attached to the diagonal. 
\item 
The origin of the coordinate system in an {\MCB}--{\DSC} plot, where $\MCB = \DSC = 0$, corresponds to the best constant forecast, namely, the unconditional event frequency in the test set.  As noted, the diagonal corresponds to its mean score, namely, $\SR = \UNC$.  Forecasts that appear above the diagonal perform better than this reference, forecast below the diagonal perform worse. 
\item
The mean score $\SC = \UNC - \DSC$ of the PAV-(re)calibrated forecast equals the (negative of the $\DSC$ component, up to a constant.  This illustrates that the forecast with the largest {\DSC} component has the greatest potential, provided (re)calibration is an option.  
\end{itemize} 

Figures \ref{fig:MCB_DSC} and \ref{fig:M1_MCB_DSC} show {\MCB}--{\DSC} plots for competing solar flare forecasts \citep{Leka2019a, Leka2019b}, and for a considerably larger number of forecasts for a binary outcome from the Fragile Families Challenge \citep{Salganik2020b}, at which we take a closer look in Section \ref{sec:FFC}.  We focus on the Brier score decomposition, which is of particular appeal, as all terms involved are guaranteed to be finite, the mean Brier score $\SX$ equals the area under the Murphy curve, and under modest conditions, the Brier score $\MCB$ component reduces to a classical measure of deviations from the diagonal in a reliability diagram \citep{Dimitriadis2021}.  Under the logarithmic score, the mean score $\SX$ equals a weighted area under the Murphy curve, and both $\SX$ and the $\MCB$ component may become infinite.  While {\MCB}--{\DSC} plots might mask details of the predictive performance, they are well suited to the task of selecting subsets of interesting forecasts that can be analyzed further by plotting triptych graphics, as exemplified in the subsequent suite of data examples. 

\section{Empirical examples}  \label{sec:empirical}

We illustrate the use of the triptych displays and  {\MCB}--{\DSC} plots for probabilistic classifiers from the academic literature in astrophysics, economics, and the social sciences.

\subsection{Solar flares}  \label{sec:solar}

\begin{figure}[p]
\centering
\includegraphics[width=\linewidth]{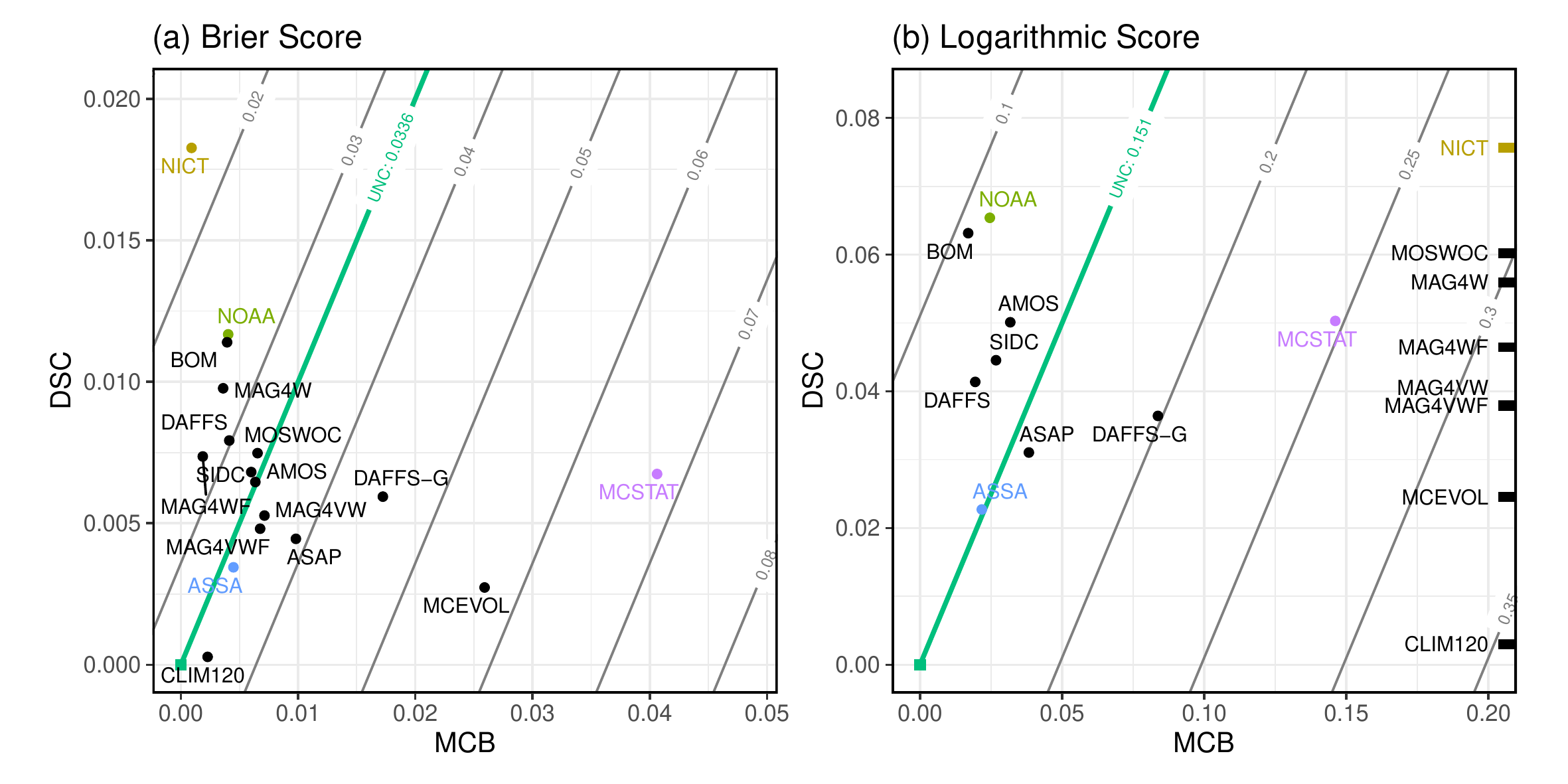}
\caption{{\MCB}--{\DSC} plots for probability forecasts of class M1.0+ solar flares under (a) the Brier score and (b) the logarithmic score.  Colors align with Figure \ref{fig:M1_triptych}.  The green square at the origin represents the ex post best constant forecast, that is, the unconditional event frequency, and the thick green line separates forecasts that are better (above the line) and that are worse (below the line) than this baseline.  Forecasts shown along the right margin in panel (b) have an infinite mean logarithmic score.}  \label{fig:M1_MCB_DSC}
\end{figure}

\begin{figure}[p]
\centering
\includegraphics[width=\linewidth]{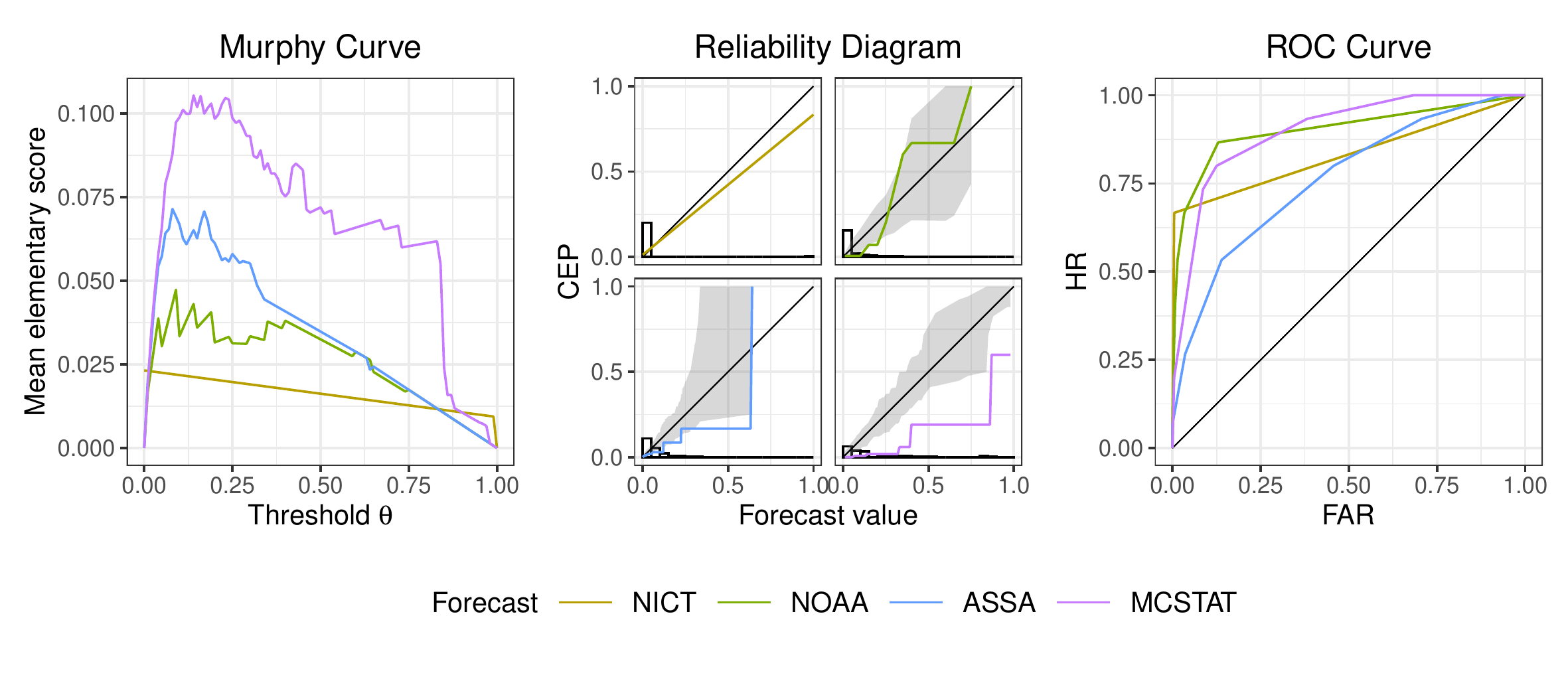}
\vspace{-10mm}
\caption{Triptych graphics for probability forecasts of class M1.0+ solar flares. Reliability curves are shown on (the smallest contiguous interval containing) the support of the forecast distribution.}  \label{fig:M1_triptych}
\end{figure}

Solar flares are energetic phenomena with potentially disastrous effects on modern terrestrial communications systems.  Owing to the increased availability of astrophysical data in real time, numerous forecasting systems for solar flares have been developed, and in a series of workshops a data repository for comparative evaluation has been created \citep{Barnes2016, Leka2019a, Leka2019b}.

We consider operational probability forecasts at a prediction horizon of a day ahead for solar flares of class C1.0+ and M1.0+, in exceedance of $10^{-6}$ and $10^{-5}$ Watts per square meter, respectively, as issued in calendar years 2016 and 2017.  While \citet{Leka2019a} and \citet{Leka2019b} describe 11 and 19 competing forecasts for C1.0+ and M1.0+ flares, respectively, there are substantial amounts of missing data in the records.  As fair comparisons require evaluation on a joint set of forecast situations, we restrict our analysis to test sets of 9 forecasts for C1.0+ flares on 577 days, as analyzed in Figures \ref{fig:C1_triptych} and \ref{fig:MCB_DSC}, and 17 forecasts for M1.0+ flares on 431 days.  On these test sets, records are complete and flares have unconditional event frequencies of 30.3 and 3.5 percent, respectively.

Turning to M1.0+ flares, Figure \ref{fig:M1_MCB_DSC} shows {\MCB}--{\DSC} plots under the Brier score and the logarithmic score.  Notably, under the logarithmic score most forecasts are outperformed by the best constant forecast.  For the triptych graphics in Figure \ref{fig:M1_triptych} we select the NOAA forecast, which performs well under either scoring rule, and the NICT forecast, which is by far the best performing method in terms of the Brier score.  Furthermore, we consider the MCSTAT forecast, which is poorly calibrated, and the ASSA forecast as a technique that lacks discrimination ability.  Due to the low unconditional event probability, most forecast values are small.  The NOAA forecast is of high quality in every regard.  The NICT forecast is a hard classifier, that is, it issues forecast probabilities of 0 and 1 only.  It performs best under most thresholds $\theta$ in the Murphy diagram, except at very high values.  Not surprisingly, it is penalized by an infinite mean logarithmic score.    The MCSTAT forecast exhibits good discrimination ability, but overpredicts, that is, the conditional event frequency is persistently lower than the forecast value, resulting in poor overall performance.  These issues can be addressed by (re)calibration, as opposed to the lack of discrimination ability of the ASSA forecast, which cannot be remedied.

\subsection{Survey of Professional Forecasters (SPF) probability forecasts of economic recessions} \label{sec:GDP}

We proceed to study probability forecasts for US GDP recessions, that is, quarters with a negative real GDP growth rate.  The data base of the Survey of Professional Forecasters \citep[SPF:][]{Croushore2019} includes probability forecasts for a GDP decline in the current and the following four quarters from the fourth quarter of 1968 through the third quarter of 2019.\footnote{SPF forecasts are available under \href{https://www.philadelphiafed.org/surveys-and-data/recess}{https://www.philadelphiafed.org/surveys-and-data/recess} and binary outcomes under \href{https://www.philadelphiafed.org/surveys-and-data/real-time-data-research/routput}{https://www.philadelphiafed.org/surveys-and-data/real-time-data-research/routput}.  Data for the third quarter of 1975 are missing.}   Following \citet{Lahiri2013}, we consider the mean over all individual SPF forecasts, which we denote SPF Consensus, and SPF forecaster \#65, who reports the second most frequently among the survey participants.  \citet{Lahiri2013} study SPF probability forecasts through the first quarter of 2011 by evaluating calibration, assessing potential predictive ability through ROC curves, and reporting mean Brier and mean logarithmic scores.  While their analysis is in the spirit of the triptych approach, it differs by necessity, as the methods proposed here depend on recent methodological advances \citep{Ehm2016, Dimitriadis2021} not yet available then.

\begin{figure}[t]
\centering
 \includegraphics[width=\linewidth]{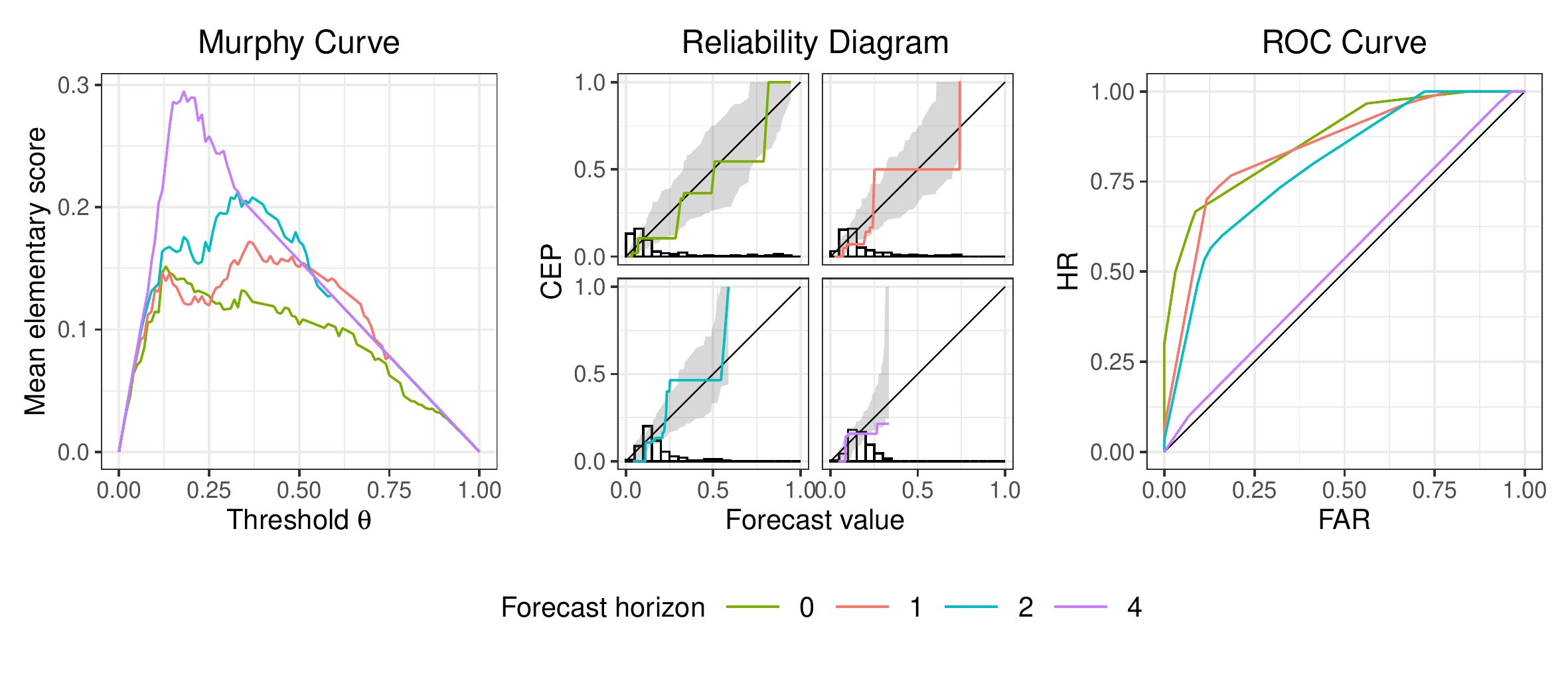}
\vspace{-10mm}
\caption{Triptych graphics for SPF Consensus forecasts of US recessions at different prediction horizons}  \label{fig:SPF_horizons}		
\end{figure}

The triptych graphics in Figure \ref{fig:SPF_horizons} serve to compare SPF Consensus forecasts at prediction horizons ranging from current quarter nowcasts to four quarters ahead.  The test set comprises 195 quarters between the second quarter of 1971 and the first quarter of 2019, with unconditional event frequency 0.16.  We note the unsurprising yet drastic effects of the forecast horizon on the quality of the SPF Consensus forecast.  While forecasts are reasonably well calibrated at all prediction horizons, the discrimination ability as shown by the ROC curves and, hence, overall predictive ability as visualized by the Murphy curves, deteriorate dramatically with the prediction horizon.  Forecasts four quarters ahead have virtually no discrimination ability.  Table \ref{tab:SPF_decom} concerns a test set of 61 quarters between 1972 and 2006, for which predictions by SPF forecaster \#65 are available, with unconditional event frequency 0.23.  The SPF Consensus forecast outperforms SPF forecaster \#65 at all prediction horizons considered.  The predictive performance of the individual forecaster falls markedly below that of the unconditional reference forecast at a prediction horizon of two quarters already.  The SPF Consensus forecast maintains superior discrimination ability and overall performance at a prediction horizon of two quarters, and performs en par with the unconditional reference forecast at a prediction horizon of four quarters ahead.

\begin{table}[t]
\centering
\footnotesize
\caption{CORP decomposition of mean Brier score for probability forecasts of US recessions from SPF Consensus and SPF \#65. The {\UNC} component equals the mean Brier score for the best constant forecast, namely, the unconditional event frequency in the test set, at 0.177.}  \label{tab:SPF_decom}
\begin{tabular}{l r ccc r ccc r ccc}
\toprule
Forecast && \multicolumn{3}{c}{$h = 1$} && \multicolumn{3}{c}{$h = 2$} && \multicolumn{3}{c}{$h = 4$} \\
\cmidrule{3-5} \cmidrule{7-9} \cmidrule{11-13} 
              && $\SX$ & {\MCB} & {\DSC} && $\SX$ & {\MCB} & {\DSC} && $\SX$ & {\MCB} & {\DSC} \\ 
\midrule
SPF Consensus && 0.118 & 0.045  & 0.104  && 0.144 & 0.043  & 0.075  && 0.177 & 0.018  & 0.018  \\
SPF \#65      && 0.143 & 0.019  & 0.053  && 0.207 & 0.043  & 0.013  && 0.212 & 0.036  & 0.001  \\
\bottomrule
\end{tabular}
\medskip
\end{table}

\subsection{Fragile Families Challenge}  \label{sec:FFC}

\begin{figure}[p]
\centering
\includegraphics[width=\linewidth]{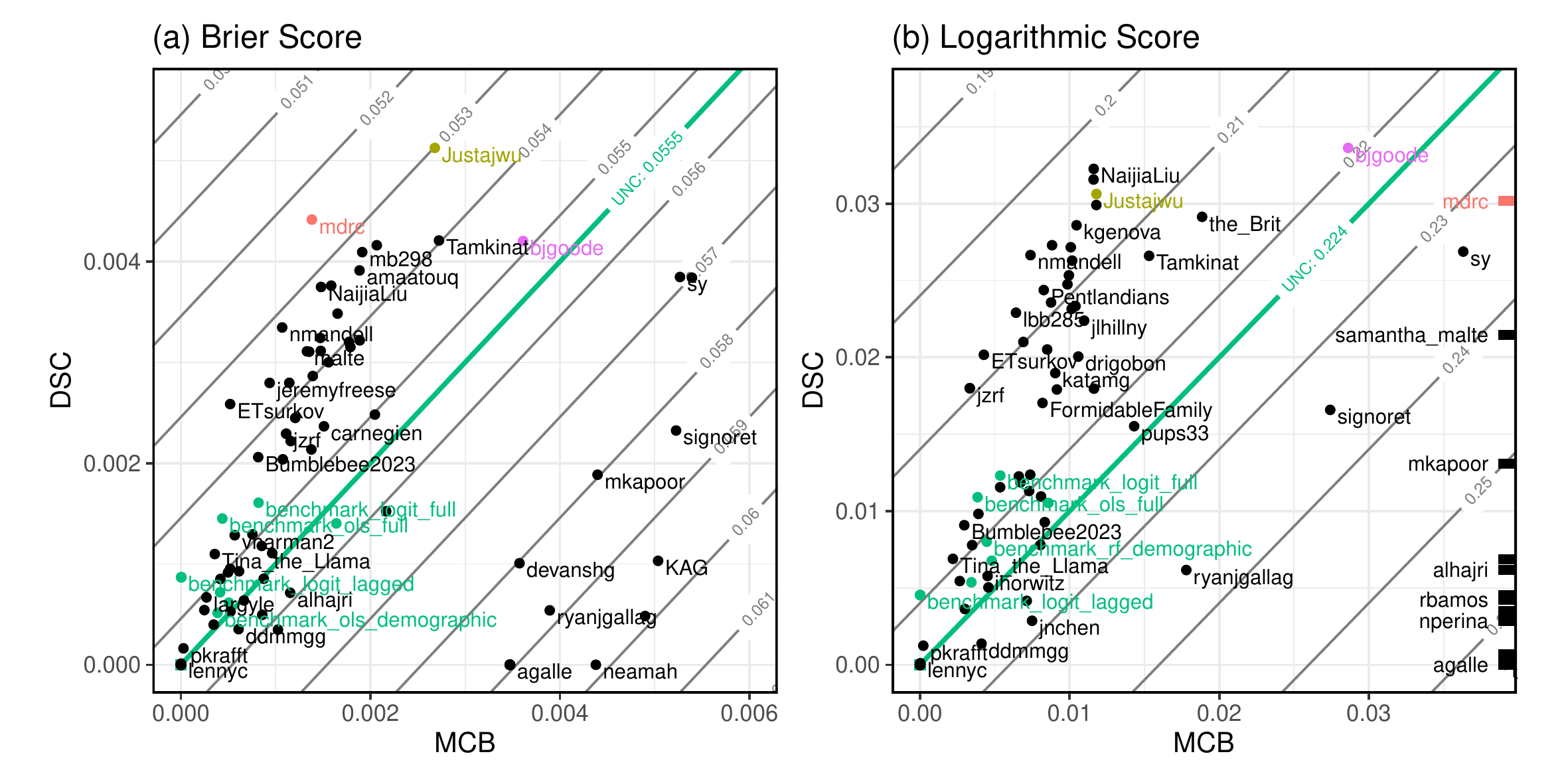}
\caption{{\MCB}--{\DSC} plots for probability forecasts of eviction from the Fragile Families Challenge under (a) the Brier score and (b) the logarithmic score.  Benchmark forecasts are marked in green; the other colors align with Figure \ref{fig:FFC_triptych}.   The green square at the origin represents the ex post best constant forecast, that is, the unconditional event frequency, and the thick green line separates forecasts that are better (above the line) and that are worse (below the line) than this baseline.  Forecasts shown along the right margin in panel (b) have infinite mean logarithmic score.  Various forecasts are not represented in the displays, due to trivial submissions \citep[Table S5]{Salganik2020b}, overlap in symbols or labels, or a particularly poor (but finite) mean score.}  \label{fig:FFC_MCB_DSC}
\end{figure} 

\begin{figure}[p]
\centering
\includegraphics[width=\linewidth]{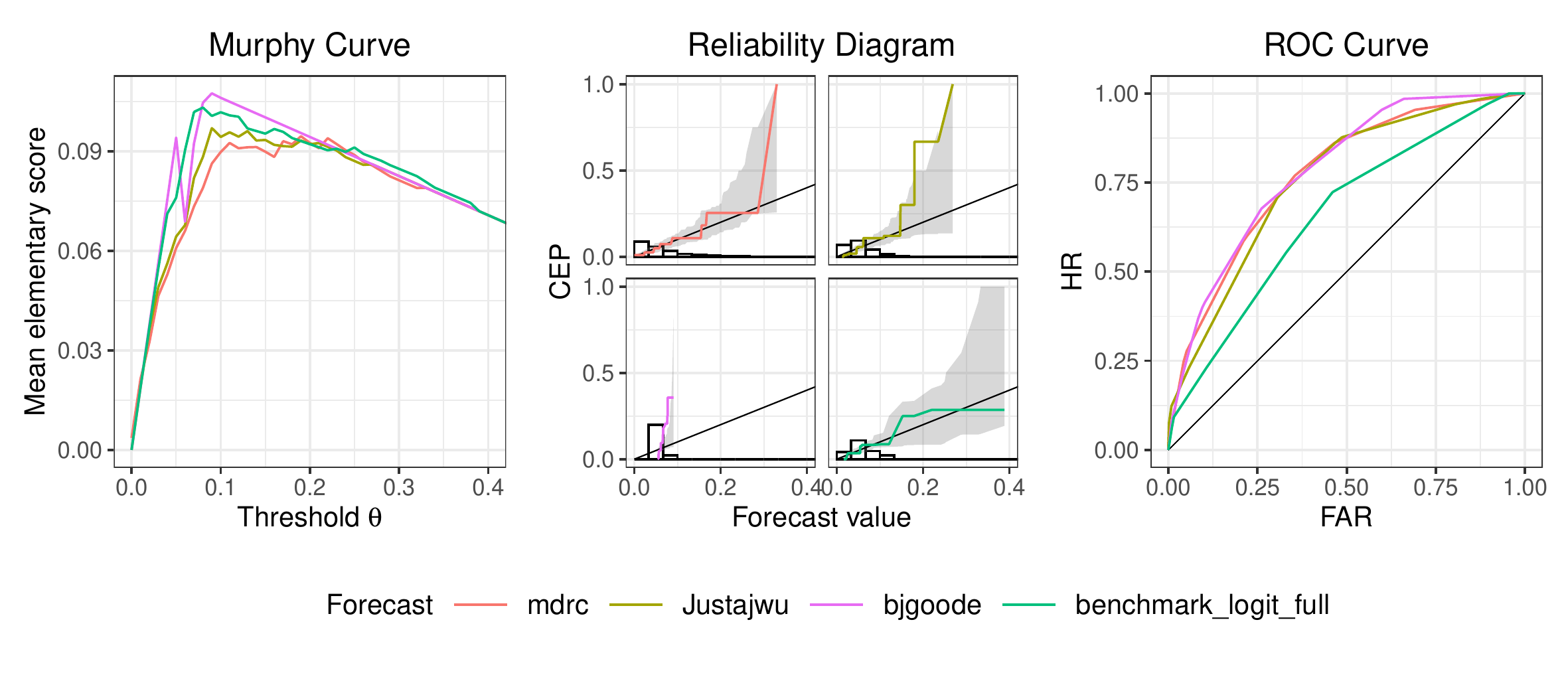}
\vspace{-10mm}
\caption{Triptych graphics for probability forecasts of eviction from the Fragile Families Challenge.}  \label{fig:FFC_triptych}
\end{figure}

The Fragile Families Challenge \citep{Salganik2020a, Salganik2020b, Salganik2021} is a scientific mass collaboration where teams supplied predictions for six (three binary and three real-valued) variables about life trajectories of children and families, based on a rich data set from the Fragile Families and Child Wellbeing Study \citep{Reichman2001}.  \citet{Salganik2020b} posit in their abstract that 
\begin{quote}
\footnotesize
''despite using a rich dataset and applying machine-learning methods optimized for prediction, the best predictions were not very accurate and were only slightly better than those from a simple benchmark model''.  
\end{quote}
We use the triptych methodology to shed detailed light on this claim for one of the binary outcomes in the study, namely, eviction (from a family's home or apartment, for not paying rent or mortgage).  For eviction, and also for the binary outcome job training (specifically, primary caregiver participation in job training) analyzed in panel (b) of Figure \ref{fig:MCB_DSC}, probability forecasts were sought for a holdout set of 1,103 families, with 160 teams providing valid contributions.  In addition, the Challenge organizers supplied nine benchmark forecasts based on commonly used statistical and machine learning techniques.  The unconditional event frequency in the holdout set is 0.059 for eviction and 0.246 for job training. 

The substantial numbers of up to 169 forecasts to be compared discourage the immediate use of the triptych graphics.  To enable the selection of methods of interest, Figure \ref{fig:FFC_MCB_DSC} shows {\MCB}--{\DSC} plots for the eviction data set under the Brier score and the logarithmic score, respectively.  Benchmark forecasts are represented in green and cluster near the origin, that is, they are well calibrated but lack discrimination.  Interestingly, a substantial number of competitors outperform the benchmarks of \citet{Salganik2020b} and the best constant forecast in terms of both scores, though the improvement is small.  However, for the job training data set in panel (b) of Figure \ref{fig:MCB_DSC} none of the teams show predictive ability superior to the benchmark forecasts.  

For the triptych graphics in Figure \ref{fig:FFC_triptych} we select mrdc as the best forecast in terms of the Brier score, bjgoode and Justajwu as the best discriminating forecasts with respect to the Brier score and the logarithmic score, respectively, and the baseline technique benchmark\_logit\_full of  \citet{Salganik2020b}.  The mrdc, bjgoode, and Justajwu forecasts outperform the baseline model in terms of discrimination ability, as depicted by the ROC curve.  Due to the low unconditional event frequency, the forecasts take on values below 0.40 only, and we restrict the reliability diagrams and Murphy curves to this range.  The baseline model is particularly well calibrated, which makes it competitive in terms of overall predictive performance, as demonstrated in the Murphy curves.  However, (re)calibrated versions of the mrdc, bjgoode, and Justajwu forecasts are bound to outperform the benchmarks by notable margins.

\section{Discussion}  \label{sec:discussion}

In this paper, we have proposed the joint use of a triptych of diagnostic graphics in the evaluation of probability forecasts, including the reliability diagram in the recently proposed CORP form to assess calibration, the concave variant of the receiver operating characteristic (ROC) curve to elucidate discrimination ability, and the Murphy curve for the overall assessment of predictive performance and economic utility.  For a succinct overview of the performance of multiple forecasts we have introduced {\MCB}--{\DSC} plots that leverage the CORP decomposition of a mean proper score into miscalibration (\MCB), discrimination (\DSC), and uncertainty (\UNC) components.  Software for the implementation of these tools and for the replication of the results in the article are available in \textsf{R} \citep{R, replication_triptych}.  An \textsf{R} package implementation is currently in development.

Our work builds on and supplements, and in a sense completes, extant software for the evaluation of probabilistic classifiers, or probabilistic forecasts in general, including but not limited to the \texttt{ROCR} \citep{Sing2005}, \texttt{pROC} \citep{Robin2011}, \texttt{Murphydiagram} \citep{Ehm2016}, \texttt{verification} \citep{R_verification}, and \texttt{reliabilitydiag} \citep{Dimitriadis2021} packages in \textsf{R}.  Arguably, closest in spirit are the \texttt{classifierplots} \citep{R_classifierplots} package, which generates a ``grid of diagnostic plots'' that includes reliability diagrams and ROC curves, and the interactive \texttt{Calibrate} approach of \citet{Xenopoulos2023}.  However, these packages do not use the CORP approach of \citet{Dimitriadis2021} for the generation of reliability diagrams and score decompositions, nor do they implement Murphy diagrams.  

In view of the general theory of calibration and score decompositions developed by \citet{Gneiting2021}, the triptych approach to the diagnostic evaluation of probability forecasts might serve as a blueprint for evaluation strategies in similar settings, including but not limited to ordinary least squares regression, forecasts in the form of the expected value of a general real-valued outcome, quantile regression, and quantile forecasts.  The case of quantiles has been studied by \citet{Gneiting2023}, whose toolbox includes variants of Murphy curves, CORP reliability diagrams, and the CORP score decomposition.  The recently developed universal ROC (UROC) curve of \citet{Gneiting2022b} generalizes the ROC curve from the classical case of a binary outcome to a general real-valued outcome, and the UROC curve might join CORP reliability diagrams and Murphy curves to form triptych graphics in the above types of settings.  While currently available implementations of the triptych graphics and {\MCB}--{\DSC} plots involve static graphics only, ever increasing numbers of competitors in forecast contests \citep{Salganik2020b, Makridakis2022} may warrant the development of interactive versions, where users can select competitors of interest in an {\MCB}--{\DSC} plot and generate the respective triptych graphics on the fly. 

\section*{Acknowledgement} 

Timo Dimitriadis, Tilmann Gneiting and Alexander Jordan gratefully acknowledge support by the Klaus Tschira Foundation.  Timo Dimitriadis gratefully acknowledges financial support from the German Research Foundation (DFG) through grant number 502572912.  The work of Peter Vogel was funded by DFG through grant number 257899354.  We thank Kajal Lahiri, Johannes Resin, Stefan Trautmann, and participants at the 2019 and 2022 International Symposium on Forecasting in Thessaloniki and Oxford, respectively, COMPSTAT 2022 in Bologna, and CFE 2022 in London for comments and advice.

\appendix

\section*{Appendix}

\section{Theoretical guarantees: Rigorous statements and proofs}  \label{app:proofs}

We work within the prediction space setting of \citep{Gneiting2013} and \citet{Gneiting2022a}, where we represent the probability forecast and the binary outcome as random variables $X$ and $Y$, respectively, with joint distribution $\myQ$, where $Y = 1$ represents an event and $Y = 0$ a non-event, with both types of outcomes having strictly positive probability.  The symbol $F$ denotes the marginal cumulative distribution function (CDF) of the forecast $X$.  When comparing competing probability forecasts $X_1$ and $X_2$ for the binary outcome $Y$, we also denote their joint distribution by $\myQ$.  

The probability forecast $X$ is \textit{calibrated} if $\myQ(Y = 1 \mid X) = X$ almost surely.  We define the conditional CDF $F_{|1}(t) = \myQ(X \leq t \mid Y=1)$  and $F_{|0}(t) = \myQ(X \leq t \mid Y=0)$ such that, for any threshold value $t \in [0,1]$, the population versions of \textit{hit rate}\/ (HR) and \textit{false alarm rate}\/ (FAR) are given by
\begin{align*}
\HR(t) = 1 - F_{|1}(t) = \myQ(X > t \mid Y = 1) \quad \textrm{and} \quad \FAR(t) = 1 - F_{|0}(t) = \myQ(X > t \mid Y = 0), 
\end{align*}	
respectively.  If $F_{|0}$ and $F_{|1}$ are continuous and strictly increasing, the ROC curve can be identified with a function $R: [0,1] \to [0,1]$, where $R(0) = 0$, $R(1) = 1$, and $R(p) = 1 - F_{|1}(F_{|0}^{-1}(1-p))$ for $p \in (0,1)$.  In the general setting, including but not limited to the case of empirical distributions for data of the form \eqref{eq:data}, the \textit{raw ROC diagnostic}\/ is the set-theoretic union of the points of the form $(\FAR(t), \HR(t))'$ within the unit square, from which the \textit{ROC curve}\/ is obtained by linear interpolation \citep{Gneiting2022a}.  On the support of $F$, we can index the ROC curve in terms of $c = F(t)$, with a natural extension to $c \in [0,1]$ via linear interpolation. 

\begin{figure}[t]
\centering
\includegraphics[width = \linewidth]{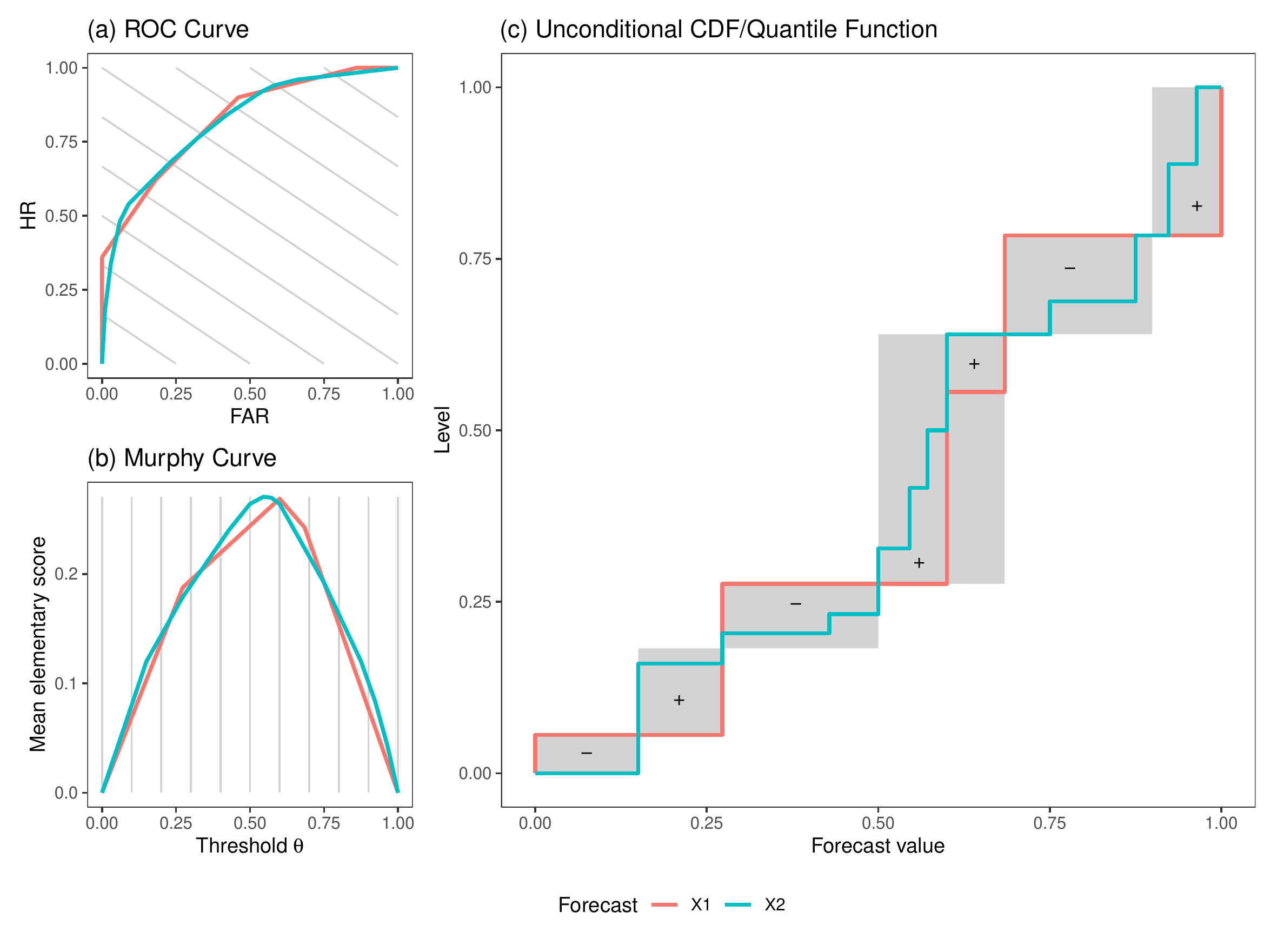}
\caption{Schematic illustration of the technical concepts in the statements and proofs of Lemma \ref{lem:crossingpoints} and Theorem \ref{thm:crossingpoints}.  Given calibrated classifiers $X_1$ and $X_2$ for the binary outcome $Y$, where $0 < \pi_0 = \myQ(Y = 0) < 1$, the panel at right shows the unconditional CDFs, $F_1$ and $F_2$, and generalized inverses or quantile functions, $Q_1$ and $Q_2$, respectively.   At lower left, the Murphy curve difference $\mathrm{D}_{{ X_1}, { X_2}}^\mathrm{MC}(\theta)$ arises as the vertical difference between the two Murphy curves at the cost--loss parameter $\theta$.  The panel at upper left concerns the ROC curve difference $\mathrm{D}_{{ X_1}, { X_2}}^\mathrm{ROC}(c)$.  The parallel lines with slope $- \pi_0/(1 - \pi_0)$ connect points on the two ROC curves with the same index, $c$.}  \label{fig:crossingpoints}
\end{figure}

Except for assuming that the competing probability forecasts $X_1$ and $X_2$ are calibrated, and requiring that $\myQ( Y = 0) \in (0,1)$, we do not impose any regularity conditions on the distribution $\myQ$.  Our proofs rely on the novel Lemma \ref{lem:crossingpoints} that expresses the difference of the Murphy curves, 
\begin{align*}   \label{eq:D_MC}
2\, \mathrm{D}_{{ X_1}, { X_2}}^\mathrm{MC}(\theta) = \myE_\myQ \, \myS_\theta(X_1,Y) - \myE_\myQ \, \myS_\theta(X_2,Y)
\end{align*} 
at the cost--loss parameter $\theta \in (0,1)$, as sketched in the lower left display of Figure \ref{fig:crossingpoints}, and the difference of the ROC curves, $\mathrm{D}_{{ X_1}, { X_2}}^\mathrm{ROC}(c)$, at $c = F(t)$ in terms of the unconditional CDFs, $F_1$ and $F_2$, and the associated quantile functions, $Q_1$ and $Q_2$, of the classifiers $X_1$ and $X_2$.  As the proof of Theorem \ref{thm:crossingpoints} demonstrates, and the upper left panel of Figure \ref{fig:crossingpoints} illustrates, the respective points on the ROC curves lie on parallel lines with slope $-\pi_0/(1 - \pi_0)$, and the ROC curve difference $\mathrm{D}_{{ X_1}, { X_2}}^\mathrm{ROC}(c)$ is taken in this direction.  We note the close relation to the idea that underlies the \textit{Kendall curve}\/ \citep{HernandezOrallo2013}, which considers the difference between a ROC curve and its optimal version, that is, the left and upper boundary of the unit square, in exactly this direction.  The Kendall curve then presents the information contained in the ROC curve in cost space, which coincides with the abscissa of the Murphy diagram.

\begin{lemma}  \label{lem:crossingpoints}
Suppose that $X_1$ and $X_2$ are calibrated probability forecasts for the binary outcome $Y$, where $\myQ(Y = 0) \in (0,1)$.  Then the Murphy curve difference and the ROC curve difference are 
\begin{equation}  \label{eq:lem}
\mathrm{D}_{X_1,X_2}^\mathrm{MC}(\theta) = \int_{[0,\theta]} \left( F_2(x) - F_1(x) \right) \mathrm{d}x 
\quad \text{and} \quad
\mathrm{D}_{X_1,X_2}^\mathrm{ROC}(c)
= \int_{[0,c]} \left( Q_1(\alpha) - Q_2(\alpha) \right) \mathrm{d}\alpha,
\end{equation}
where $F_1$ and $F_2$ are the unconditional CDFs for $X_1$ and $X_2$, and $Q_1$ and $Q_2$ are the left-continuous generalized inverses to $F_1$ and $F_2$, respectively.
\end{lemma}

\begin{proof}
To simplify notation, let $X$ be a probability forecast for $Y$ with unconditional CDF $F$ and generalized inverse $Q$.  Let $F_{|i} = \myQ(X \leq \cdot \mid Y = i)$ and $\pi_i = \myQ(Y = i)$ for $i \in \{0, 1\}$.  

For the Murphy curve difference, consider the elementary scoring function $\myS_\theta$ from \eqref{eq:elementary}.  For $\theta \in (0,1)$ the expected elementary score of $X$ is
\begin{align*} 
\myE_\myQ \myS_\theta(X, Y) 
= 2 \theta \, \myQ(X > \theta, Y = 0) + 2 (1 - \theta) \left(\myQ(X < \theta, Y = 1) + \theta\, \myQ(X = \theta)\right).
\end{align*}
The Murphy curve is the graph of the map $\mathrm{MC} \colon \theta \mapsto \myE_\myQ \myS_{\theta}(X, Y)$.  Under calibration of $X$, we have $x = \myQ(Y = 1 \mid X = x)$ and $\theta\, \myQ(X = \theta) = \myQ(X = \theta, Y = 1)$, which yields
\begin{align*}
\mathrm{MC}(\theta) 
& = 2 \theta\, \myQ(Y = 0) - 2 \theta \, \myQ(X \leq \theta) + 2 \myQ(X \leq \theta, Y = 1) \\
& = 2 \theta \pi_0 - 2 \theta F(\theta) + 2 \int_{[0, \theta]} x \, \mathrm{d}F(x) = 2 \theta \pi_0 - 2 \int_{[0, \theta]} F(x) \, \mathrm{d}x
\end{align*}
using integration by parts for Lebesgue--Stieltjes integrals.  Therefore, a version of the Murphy curve difference of two forecasts ${ X_1}$ and ${ X_2}$ is 
\[
\mathrm{D}_{X_1,X_2}^\mathrm{MC}(\theta) 
= \tfrac{1}{2}\mathrm{MC}_{X_1}(\theta) - \tfrac{1}{2}\mathrm{MC}_{X_2}(\theta) 
= \int_{[0,\theta]} \left( F_2(x) - F_1(x) \right) \mathrm{d}x
\]
for $\theta \in [0,1]$, as claimed. 

Turning to the ROC curve distance, the raw ROC diagnostic of the probability forecast $X$ is the set-theoretic union of the points of the form $(\FAR(t), \HR(t))' = (1 - F_{|0}(t), 1 - F_{|1}(t))'$ for all threshold values $t$.  Under calibration of $X$, and using substitution for Lebesgue-Stieltjes integrals, we have 
\[
\pi_1 F_{|1}(t) = \int_{[0,t]} x \, \mathrm{d}F(x) = \int_{[0,F(t)]} Q(\alpha) \, \mathrm{d}\alpha.  
\]
Furthermore, $(F \circ Q \circ F)(\theta) = F(\theta) = \pi_0 F_{|0}(\theta) + \pi_1 F_{|1}(\theta)$, for $Q$ is a generalized inverse to $F$.  We use the idea that underlies the construction of the rate-driven cost curve and the Kendall curve \citep{HernandezOrallo2013}, namely, to substitute $Q(c)$ for $t$, where $c = F(t)$, in concert with the above facts, to write the points in the raw ROC diagnostic as
\begin{align*}
\left( 1 - \frac{1}{\pi_0} \left( c - \int_{[0,c]} Q(\alpha) \, \mathrm{d}\alpha \right) \! , 
       1 - \frac{1}{\pi_1} \int_{[0,c]} Q(\alpha) \, \mathrm{d}\alpha \right)',     
\end{align*}  
where $c \in \mathrm{Im}(F) = \{ F(t) \colon t \in \mathbb{R} \}$.  These expressions interpolate linearly when $c \in [0, 1] \setminus \mathrm{Im}(F)$, since $Q(\alpha) = Q(\min ( \mathrm{Im}(F) \cap [\alpha, 1]))$.  Therefore, the ROC curve as a linear interpolation of the raw ROC diagnostic is the graph of the map 
\[
\mathrm{ROC}_X^\mathrm{curve} \colon [0, 1] \to [0, 1]^2, \quad c \mapsto 
\left( 1 - \frac{1}{\pi_0} \left( c - \int_{[0,c]} Q(\alpha) \, \mathrm{d}\alpha \right) \! , 
       1 - \frac{1}{\pi_1} \int_{[0,c]} Q(\alpha) \, \mathrm{d}\alpha \right)'.    
\]
Given competing probability forecasts $X_1$ and $X_2$, the vector-valued difference between the ROC curves at $c \in [0,1]$ is  
\begin{align*}
\mathrm{ROC}_{X_1}^\mathrm{curve}(c) - \mathrm{ROC}_{X_2}^\mathrm{curve}(c) 
= \left( \frac{1}{\pi_0}, - \frac{1}{\pi_1} \right)' \left( \int_{[0,c]} \left( Q_1(\alpha) - Q_2(\alpha) \right) \mathrm{d}\alpha \right),
\end{align*}
which demonstrates that pointwise differences between ROC curves are to be measured along lines with slope $-\pi_0/\pi_1$ in the ROC curve plot, as illustrated in the upper left panel of Figure \ref{fig:crossingpoints}.  The factor at right is chosen as the pointwise distance between the ROC curves for $X_1$ and $X_2$ at index $c \in [0,1]$, that is, $\mathrm{D}_{{X_1}, {X_2}}^\mathrm{ROC}(c) = \int_{[0,c]} \left( Q_1(\alpha) - Q_2(\alpha) \right) \mathrm{d}\alpha$, as claimed.
\end{proof}

A function on the unit interval has $n$ \textit{sign changes}\/ if there exists a partition of the unit interval with $n + 1$ members, such that the function is nonnegative (nonpositive) with at least one nonzero value on the first and nonpositive (nonnegative) with at least one nonzero value on the second of any two consecutive members of the partition.  

We proceed to state and prove a rigorous version of Fact \ref{fact:crossingpoints} in Section \ref{sec:theory}. 

\begin{theorem}  \label{thm:crossingpoints}
Suppose that $X_1$ and $X_2$ are calibrated probability forecasts for the binary outcome $Y$, where $\myQ(Y = 0) \in (0,1)$.  Let $F_1$ and $F_2$ denote the unconditional CDFs of $X_1$ and $X_2$, and suppose that $F_1 - F_2$ has a finite number $n \ge 1$ of sign changes.  Then the following hold:
\begin{enumerate}
\item[(a)] The ROC curve difference $\mathrm{D}_{X_1, X_2}^\mathrm{ROC}$ has at most $n - 1$ sign changes.
\item[(b)] The Murphy curve difference $\mathrm{D}_{X_1, X_2}^\mathrm{MC}$ and the ROC curve difference $\mathrm{D}_{X_1, X_2}^\mathrm{ROC}$ have the same number of sign changes.
\end{enumerate}
\end{theorem}

\begin{proof}
For part (a), note that the integrand $Q_1 - Q_2$ of $\mathrm{D}_{X_1, X_2}^\mathrm{ROC}$ has the same number of sign changes as $F_1 - F_2$, and that no additional sign change can be introduced by integration.  Since $\mathrm{D}_{X_1, X_2}^\mathrm{ROC}(c)$ evaluates to 0 at both $c = 0$ (by definition) and $c = 1$ (due to calibration, which implies that $\myE X_1 = \myE X_2$), the number of sign changes of the integral must be smaller than the number of sign changes of the integrand.

For part (b), consider the two partitions of the unit interval generated by sign changes of the integrands $F_2 - F_1$ and $Q_1 - Q_2$, respectively.  As both partitions have the same number of elements, the elements can be matched pairwise to create blocks as illustrated in the right panel of Figure \ref{fig:crossingpoints}.  At the beginning and end of every block (the bottom left or top right corner, respectively), we have equality of the differences $\mathrm{D}_{X_1, X_2}^\mathrm{MC}(\theta)$ and $\mathrm{D}_{X_1, X_2}^\mathrm{ROC}(c)$, and within a block either both differences are nonincreasing or both are nondecreasing.  Therefore, in a single block, either both differences experience a single sign change or neither does.
\end{proof}

The assumption of finitely many sign changes in $F_1 - F_2$ is not particularly restrictive, as it is satisfied whenever either $X_1$ or $X_2$ has finite support (covering all empirical cases), or when both have a finite number of (potentially interval-valued) modes.  Otherwise, the statement in part (b) continues to hold whenever the number of sign changes in either $\mathrm{D}_{X_1, X_2}^\mathrm{MC}$ or $\mathrm{D}_{X_1, X_2}^\mathrm{ROC}$ is finite.  Furthermore, the assumption of calibration guarantees the existence of at least one sign change in $F_1 - F_2$ whenever $F_1 \neq F_2$, since otherwise $F_1$ and $F_2$ are stochastically ordered, which implies $\myE X_1 \neq \myE X_2$ and contradicts the assumption of calibration.

Informally, a probability forecast for a binary outcome is sharper than another if its forecast values are closer to the most confident values of 0 and 1, respectively \citep{Gneiting2008a}.  In order to formalize the notion of sharpness we follow \citet{Kruger2021} and define $X_1$ to be \textit{sharper}\/ than $X_2$ if it is greater in convex order, that is, if $\myE \, \phi(X_1) \geq \myE \, \phi(X_2)$ for all convex functions $\phi$ on the unit interval, with strict inequality for some $\phi$.  As \citet[p.~974]{Kruger2021} note, given the assumption that forecasts are calibrated, being larger in convex order implies indeed that the forecast values are more spread out toward the most confident probabilities of 0 and 1.  

In Fact \ref{fact:sharper} in Section \ref{sec:theory} we express the idea that if competing probability forecasts are calibrated, and one of them is sharper than the other, then the sharper one is superior in terms of both ROC curves and Murphy curves, and vice versa.  To state a rigorous version of this fact, we say that $X_1$ \textit{dominates}\/ $X_2$ \textit{in the ROC sense}\/ if $\mathrm{D}_{X_1,X_2}^\mathrm{ROC}(c) \leq 0$ for all $c \in [0, 1]$ with strict inequality at some $c$.  Similarly, $X_1$ \textit{dominates}\/ $X_2$ \textit{in the Murphy sense}\/ if $\mathrm{D}_{X_1,X_2}^\mathrm{MC}(\theta) \leq 0$ for all $\theta \in [0, 1]$ with strict inequality at some $\theta$.

\begin{theorem}  \label{thm:sharper}
Suppose that $X_1$ and $X_2$ are calibrated probability forecasts for the binary outcome $Y$, where $\myQ(Y = 0) \in (0,1)$.  Then the following relations are equivalent: 
\begin{enumerate}
\item[(i)] $X_1$ is sharper than $X_2$.
\item[(ii)] $X_1$ dominates $X_2$ in the ROC sense.
\item[(iii)] $X_1$ dominates $X_2$ in the Murphy sense.
\end{enumerate}
\end{theorem}

\begin{proof} 
The equivalence of statements (i) and (iii) is a special case of Theorem 3.1 in \citet{Kruger2021}, and the equivalence of statements (ii) and (iii) follows from the arguments in the proof of part (b) of Theorem \ref{thm:crossingpoints}. 
\end{proof}

Theorem \ref{thm:sharper} demonstrates that if probability forecasts are calibrated, then comparisons in terms of sharpness, discrimination ability, and proper scoring rules yield congruent insights, as illustrated in Scenario C and part (c) of Figure \ref{fig:ABC} in Section \ref{sec:theory}.  Related results have been discussed by \citet[p.~670]{Krzysztofowicz1990}, \citet[p.~416]{Wilks2019}, and references therein.

\addcontentsline{toc}{section}{References}

\setstretch{0.8}
\setlength{\bibsep}{5pt}

\end{document}